\documentclass[USenglish,oneside,twocolumn]{article}

\usepackage[utf8]{inputenc}
\usepackage[big]{dgruyter_NEW}

\usepackage{amsmath}
\usepackage{amssymb}

\usepackage[T1]{fontenc}    
\usepackage{hyperref}       
\usepackage{url}            
\usepackage{booktabs}       
\usepackage{nicefrac}       
\usepackage{microtype}      

\usepackage{float}

\usepackage{hyphenat}
\usepackage{xspace}

\usepackage{algorithm}
\usepackage[noend]{algpseudocode}

\makeatletter
\def\plist@algorithm{Alg.\space}
\makeatother

\usepackage{enumitem}

\usepackage{amsthm}
\theoremstyle{plain}
\newcounter{thmcounter}
\newtheorem{theorem}[thmcounter]{Theorem}
\newtheorem{lemma}[thmcounter]{Lemma}
\newtheorem{definition}[thmcounter]{Definition}
\newtheorem{corollary}[thmcounter]{Corollary}

\newcommand{\hide}[1]{}

\newcommand{\xhdr}[1]{\vspace{1.7mm}\noindent{{\bf #1.}}}

\newcommand{\Secref}[1]{Sec.~\ref{#1}}

\newcommand{\eqnref}[1]{Eq.~\ref{#1}}

\newcommand{\Thmref}[1]{Thm.~\ref{#1}}

\newcommand{\Lemmaref}[1]{Lemma~\ref{#1}}

\newcommand{\Tabref}[1]{Table~\ref{#1}}

\newcommand{\Figref}[1]{Fig.~\ref{#1}}

\newcommand{\Appref}[1]{Appendix~\ref{#1}}

\newcommand{\Algref}[1]{Alg.~\ref{#1}}

\usepackage{appendix}

\usepackage{bbm} 
\usepackage{bm}
\usepackage{mathtools}
\DeclarePairedDelimiter\abs{\lvert}{\rvert}
\DeclarePairedDelimiter\norm{\lVert}{\rVert}

\usepackage{amsmath}

\DeclareMathOperator{\sign}{sign}

\DeclareMathOperator{\Lap}{Lap}

\DeclareMathOperator*{\argmin}{arg\,min}

\usepackage[normalem]{ulem} 

\newcommand{\R}{\mathbb{R}}

\newcommand{\defeq}{\mathrel{\vcentcolon=}}

\DeclareMathOperator{\Range}{Range}

\newcommand{\x}{\bm{x}}
\newcommand{\xp}{\bm{x'}}

\usepackage{cryptocode}
\usepackage{wrapfig}

\newcommand\restr[2]{{
  \left.\kern-\nulldelimiterspace 
  #1 
  \right|_{#2} 
  }}

\DeclareMathOperator{\odpc}{ODPComposition}
\DeclareMathOperator{\odpfc}{ODPFilterComposition}
\DeclareMathOperator{\im}{M_{\text{iter}}}

\DeclareMathOperator{\comp}{COMP}
\DeclareMathOperator{\optdelta}{OptDel}
\DeclareMathOperator{\ERMOutput}{ERMOutputPert}
\DeclareMathOperator{\LogRegTest}{LogRegWithTest}

\newcommand{\pluseq}{\mathrel{+}=}
\newcommand{\minuseq}{\mathrel{-}=}

\usepackage{csvsimple}
\usepackage{booktabs,subcaption}

\begin{document}

  \author*[1]{Valentin Hartmann}

  \author[2]{Vincent Bindschaedler}

  \author[3]{Alexander Bentkamp}

  \author[4]{Robert West}

  \affil[1]{EPFL, E-mail: \url{valentin.hartmann@epfl.ch}}

  \affil[2]{University of Florida, E-mail: \url{vbindsch@cise.ufl.edu}}

  \affil[3]{State Key Laboratory of Computer Science, Institute of Software, Chinese Academy of Sciences \& Vrije Universiteit Amsterdam, E-mail: \url{bentkamp@ios.ac.cn}}

  \affil[4]{EPFL, E-mail: \url{robert.west@epfl.ch}}

  \title{\huge Privacy accounting \(\varepsilon\)conomics: Improving differential privacy composition via a posteriori bounds}

  \runningtitle{Privacy accounting \(\varepsilon\)conomics}

   \begin{abstract} 
 {Differential privacy (DP) is a widely used notion for reasoning about privacy when publishing aggregate data. 
In this paper, we observe that certain DP mechanisms are amenable to {\em a posteriori} privacy analysis that exploits the fact that some outputs leak less information about the input database than others. 
To exploit this phenomenon, we introduce output differential privacy (ODP) and a new composition experiment, and leverage these new constructs to obtain significant privacy budget savings and improved privacy--utility tradeoffs under composition.
All of this comes at no cost in terms of privacy; we do not weaken the privacy guarantee.

To demonstrate the applicability of our {\em a posteriori} privacy analysis techniques, we analyze two well-known mechanisms: the Sparse Vector Technique and the Propose-Test-Release framework. We then show how our techniques can be used to save privacy budget in more general contexts: when a differentially private iterative mechanism terminates before its maximal number of iterations is reached, and when the output of a DP mechanism provides unsatisfactory utility. Examples of the former include iterative optimization algorithms, whereas examples of the latter include training a machine learning model with a large generalization error. Our techniques can be applied beyond the current paper to refine the analysis of existing DP mechanisms or guide the design of future mechanisms.}
  \end{abstract}
  



\maketitle

\section{Introduction}
Differential privacy (DP) is a formal notion of privacy for aggregate data releases from databases. Its definition characterizes the extent of what the output of a randomized aggregation mechanism \(M\) that is invoked on a database reveals about individual database records. The guarantee is given in terms of the indistinguishability of neighboring databases, that is, databases \(\x\) and \(\xp\) where \(\x\) can be obtained from \(\xp\) by either adding/removing one record or by changing the values of one record:
\begin{definition}[Differential Privacy \cite{dwork2006calibrating,dwork2006our}]\label{def:dp}
A randomized algorithm \(M\) is \emph{\((\varepsilon,\delta)\)\hyp differentially private} if for all pairs of neighboring databases \(\x, \xp\in\mathcal{D}\) and for all \(S\subseteq \Range(M)\),
\begin{equation*}
    \Pr(M(\x)\in S)\leq e^{\varepsilon} \Pr(M(\xp)\in S) + \delta.
\end{equation*}
\end{definition}
In the definition, \(\mathcal{D}\) denotes the space of databases. If \(\delta=0\), we call \(M\) an \emph{\(\varepsilon\)\hyp differentially private} mechanism and say that \(M\) fulfills \emph{pure differential privacy}.

The DP guarantee holds over all possible sets of outputs. There is a good reason for this: we do not want to end up in a situation where we get unlucky with the database or the randomness of the mechanism \(M\) and leak more information about a database record than we intended to when releasing the output of \(M\). However, in this paper we show that when composing mechanisms, i.e., invoking a sequence of mechanisms instead of a single one, one can exploit DP guarantees that only hold w.r.t.\ proper subsets of \(\Range(M)\). We capture the collection of these subset\hyp specific guarantees in the \emph{output differential privacy} (ODP) guarantee of \(M\), which consists of a partition of \(\Range(M)\) and privacy guarantees associated with each set in the partition. By adapting mechanisms later in the sequence to the privacy guarantees associated with the outputs of previously invoked mechanisms, one can improve utility in two ways: (1) by reducing the amount of noise that is required for guaranteeing privacy, or (2) by increasing the number of mechanisms that are invoked. All of this is achieved while retaining the same standard DP guarantee for the sequence of mechanisms.

We emphasize and expand on this crucial last point: in this paper we do {\em not} weaken the DP guarantee, we do {\em not} use or propose a relaxed definition of privacy, and we do {\em not} leak any additional private information by applying our techniques. In fact, all the mechanisms we consider satisfy the traditional definition of DP (Def.~\ref{def:dp}). Instead, we simply observe that for some DP mechanisms, some outputs happen to leak less private information than other outputs. We show how to exploit this fact to improve the privacy--utility tradeoff offered by these mechanisms in practice.

To make the concept of ODP more concrete, we start with an example. Let \(f:\mathcal{D}\rightarrow\R\) be a function that maps databases to values in \(\R\) and has sensitivity \(1\), i.e., \(\abs{f(\x)-f(\xp)}\leq 1\) for all neighboring databases \(\x,\xp\in\mathcal{D}\). Then the Laplace mechanism that releases \(f(\x) + \Lap(1/\varepsilon)\), the value of \(f\) plus noise drawn from the Laplace distribution with parameter \(1/\varepsilon\), fulfills \(\varepsilon\)\hyp DP \cite{dwork2006calibrating}. With the Laplace mechanism as a building block we define the toy mechanism \(M_{\text{toy}}\) that takes as input a database \(\x\) as follows:
\begin{enumerate}
    \item Flip an unbiased coin \(c\).
    \item If \(c\) came up heads, return \(f(\x) + \Lap(1/\varepsilon)\).
    \item If \(c\) came up tails, return \(\bot\).
\end{enumerate}
Here \(\bot\) is a symbol that is independent of \(\x\).
If \(c\) comes up heads, the \(\varepsilon\)\hyp differentially private Laplace mechanism is invoked, whose output depends on \(\x\) and might contain (a limited amount of) information about individual database records.
If \(c\) comes up tails, however, the output is independent of \(\x\) and thus does not contain any information about individual records in \(\x\). The fact that \(c\) comes up tails also does not reveal any information about \(\x\) since the coin flip is independent of \(\x\) as well. Thus, an adversary learns nothing about \(\x\) if they receive \(\bot\) as the output of \(M_{\text{toy}}\). This means that in the case of a \(\bot\)\hyp output the adversary should be allowed to receive the result of a second \(\varepsilon\)\hyp differentially private mechanism if the overall privacy budget is \(\varepsilon\). In the case where an output is produced via the Laplace mechanism, however, the adversary should not receive a second output.

While this is a toy example where the output \(\bot\) serves no practical purpose, we show examples of well\hyp known mechanisms that exhibit the same behavior --- some outputs leak more private information than others --- notably the Sparse Vector Technique \cite{dwork2009complexity,svt2017} and the mechanisms from the Propose-Test-Release framework \cite{dwork2009differential}.

\subsection{Our contributions}
In this paper we introduce the concept of \emph{output differential privacy} (ODP), which can be used to more accurately describe the leakage of private information of mechanisms whose different outputs reveal different amounts of information about the database they are invoked on. Since ODP is an extension of DP, there exists a trivial ODP guarantee for every DP mechanism. However, our framework only yields improvements for mechanisms with non\hyp trivial ODP guarantees. This class of mechanisms includes the well\hyp known Sparse Vector Technique (SVT) and the mechanisms from the Propose\hyp Test\hyp Release (PTR) framework, but also mechanisms that can be derived from DP mechanisms with only trivial ODP guarantees (\Secref{sec:iteration_length} and \ref{sec:recovering_budget}), even in a black\hyp box fashion (\Secref{sec:iteration_length}). When composing mechanisms with non\hyp trivial ODP guarantees with other DP mechanisms, the more fine\hyp grained ODP guarantees can be used to improve the utility of the composition over using the coarse DP guarantees. Utility here is measured in terms of the noise required to be added or the maximal number of allowed mechanism invocations to not exceed a given DP guarantee. This is achieved via a novel composition protocol that keeps track of the actual leakage of the mechanism's outputs instead of using the leakage of the worst\hyp case output, in a way that preserves standard DP.

\xhdr{How to benefit from ODP}
For simplicity assume that we only compose one ODP mechanism \(M_1\) with one other DP mechanism \(M_2\). After having produced an output \(s_1\) via \(M_1\), we check how much \(s_1\) would reveal about the worst\hyp case database that \(M_1\) could have been invoked on. For a worst\hyp case output this bound will not be better than the regular DP bound. For a non\hyp worst\hyp case output such as the \(\bot\) from the example of \(M_{\text{toy}}\), however, we have a better bound on the leakage than the DP bound and need to subtract less from the remaining privacy budget. This means that there is more privacy budget left to spend on \(M_2\), and hence \(M_2\) can produce a less noisy and more accurate result. We give some concrete examples for this in \Secref{sec:analysis_existing}.
Alternatively, we could decide to spend the saved privacy budget on invoking a third mechanism \(M_3\).
In some cases it might not be of interest to invoke more than one mechanism on the database, e.g., when the database serves only as training data for one particular machine learning (ML) model. In such cases saved privacy budget can be spent on mechanism invocations on other databases to which individuals from the first database might have contributed as well. Examples for this are user databases for different products of the same company or the results of different surveys from the same city.

\xhdr{Useful for small to medium length compositions}
There has been a lot of research on better asymptotic composition bounds (advanced composition bounds) when the number of composed mechanisms is large (see \Secref{sec:related}), whereas our framework yields improvements for as little as two composed mechanisms, up to a small to medium number of mechanisms. For a more thorough discussion, see \Appref{sec:adv_comp}. There we also discuss the potential for an advanced composition theorem within our framework.

\xhdr{A formally verified composition theorem}
Motivated by mistakes in previous composition theorems \cite{svt2017,rogers2016privacy}, we have formally verified the proof of our ODP composition theorem in the proof assistant Lean \cite{demoura2015lean}.
Proof assistants are software tools
to develop and check formal mathematical proofs. They are used in various areas of computer science --- e.g., to verify
algorithms and data structures, programming language semantics, security protocols, or hardware specifications. We make the formal proof available online and hope that it can help future formalization endeavors of DP mechanisms or theorems.
In this paper we are not concerned with measurability and assume that all sets that we deal with are measurable. However, in the formal proof of the composition theorem we also show measurability.

\subsection{Organization of the paper}
We start by summarizing prior work and describing how it relates to ODP in \Secref{sec:related}. We then formally introduce our ODP framework, consisting of definitions and a novel composition protocol, in \Secref{sec:method}. The privacy proof of the composition protocol is deferred to the appendix. As already mentioned, the ODP framework can be applied to PTR and the SVT, which we describe in detail in \Secref{sec:analysis_existing}.
In \Secref{sec:iteration_length} we demonstrate how ODP can be used to save privacy budget in iterative mechanisms with a non\hyp fixed number of iterations. ODP also allows to recover already spent privacy budget in case the output of a DP mechanism is unsatisfactory, as we show in \Secref{sec:recovering_budget}.
We conclude in \Secref{sec:conclusions}.
The appendix contains the discussion of a possible extension to advanced composition (\Appref{sec:adv_comp}) and most of the proofs.

\section{Related work}
\label{sec:related}
One of the two major components of the ODP framework is a new composition experiment that improves utility in certain cases.
Extending privacy guarantees from a single mechanism to sequences of mechanisms has been an important subject of study from the early days of DP research. The first result for the composition of mechanisms is the simple composition theorem \cite{dwork2006our}, which states that a mechanism that invokes \(k\) \((\varepsilon,\delta)\)\hyp DP mechanisms fulfills \((k\varepsilon,k\delta)\)\hyp DP. This statement also holds for pure DP with \(\delta=0\) and cannot be improved upon if pure DP is also required for the composition. However, Dwork et al.\ \cite{dwork2010boosting} later proved the advanced composition theorem, which shows that for the cost of an increase in the \(\delta\)\hyp part of the composition guarantee, the \(\varepsilon\)\hyp part can be decreased to \(\mathcal{O}(\varepsilon^2 k + \varepsilon\sqrt{k})\). Since then, optimal composition theorems both for the case of homogeneous composition \cite{kairouz2017composition} (where all composed mechanisms have the same \((\varepsilon,\delta)\)\hyp guarantee) and heterogeneous composition \cite{murtagh2016complexity} (where they may have different \((\varepsilon,\delta)\)\hyp guarantees) have been found.
The optimality of these composition theorems holds w.r.t.\ general \((\varepsilon,\delta)\)\hyp DP mechanisms, that is, if the DP guarantees of the mechanisms are fixed, but the data analyst is free to choose any mechanisms that fulfill these DP guarantees. By restricting the choice of mechanisms, tighter composition bounds can be given. To this end, relaxations of DP such as concentrated DP \cite{dwork2016concentrated}, the Rényi\hyp divergence based zero-concentrated DP \cite{bun2016concentrated} and Rényi DP \cite{mironov2017renyi}, and \(f\)\hyp DP \cite{dong2019gaussian} have been introduced. These definitions can capture the composition behavior of specific mechanisms more precisely.
What these improvements have in common with advanced composition is that they only give an improved level of privacy over simple composition if the number of composed mechanisms is large enough (see \Tabref{table:composition_cutoff}).

The most flexible setting that classic composition theorems consider is one where the number of mechanisms to be invoked and their DP guarantees are fixed ahead of time, but where under these constraints the data analyst may adaptively choose the mechanism to be invoked in each step and the database to invoke it on based on the outputs of the previous mechanisms. This is formalized in a so\hyp called composition experiment \cite{dwork2010boosting}.
The ODP composition experiment that we propose gives the data analyst a level of freedom that goes beyond what is possible in the classic setting: the data analyst may adaptively choose not only mechanisms and databases, but also neither the number of mechanisms nor their DP guarantees need to be fixed ahead of time. This setting has been previously studied by Rogers et al.\ \cite{rogers2016privacy}. They introduce privacy filters, which are functions that can be used as stopping rules to prevent a given privacy budget from being exceeded. This is essentially the same as how we prevent the choice of mechanisms whose invocation would exceed the privacy budget in our composition experiment. In fact, we can even almost equivalently reformulate our composition experiment using a new variant of privacy filters (see \Appref{sec:adv_comp}). As opposed to us, Rogers et al.\ give advanced composition results that are asymptotically (in the number of queries) better than our simple composition results. These results have been improved upon by showing that composition results obtained via Rényi DP also hold for the privacy filter setting \cite{feldman2021individual,lecuyer2021practical}. However, our composition results are still superior whenever the number of queries is small to medium.

The observation that some outputs of DP mechanisms leak less private information than others and that in these cases one only has to account for the leakage of the actual output has been exploited before, though only in the design of quite specific types of mechanisms and not for developing a general framework as we do. Propose\hyp Test\hyp Release (PTR) \cite{dwork2009differential} is a method for designing DP mechanisms based on robust statistics. PTR mechanisms consist of chains of mechanisms of a particular type whose different outputs leak different amounts of private information, and the authors exploit the fact that only certain sequences of outputs are possible to give better DP guarantees. The Sparse Vector Technique (SVT) \cite{dwork2009complexity,svt2017} is a technique for releasing the (binary) results of a sequence of threshold queries with DP, in a way that each positive result contributes a certain amount to the private leakage, but all negative results together only contribute a fixed amount, which can be used to output arbitrarily many negative results with a fixed privacy budget. The DP mechanism for top\hyp \(k\) selection by Durfee and Rogers \cite{durfee2019practical} can be seen as a combination of the ideas of PTR and the SVT. The goal of their mechanism is to return the top\hyp \(k\) elements from a database. However, the mechanism might return less than \(k\) elements. In that case, it can be invoked again multiple times until \(k\) elements have been returned, with an additional cost in \(\delta\) but without any additional cost in \(\varepsilon\).
We dedicate \Secref{sec:analysis_existing} to PTR and the SVT, where we show how with ODP we can reduce the privacy budget that these mechanisms use up when composing them with other mechanisms.

Dwork and Rothblum \cite{dwork2016concentrated} formalize the distribution of the amount of leakage of private information of mechanisms over their outputs via the so\hyp called privacy loss random variable. They --- and later Sommer et al.\ \cite{sommer2019privacy} --- use the fact that it is unlikely that a mechanism will, over many iterations, always produce an output with a high privacy loss to show improved composition bounds. These are, however, {\em a priori} bounds that do not take into account the privacy loss of the actually produced outputs.

Ligett et al.\ \cite{ligett2019accuracy} do compute the privacy loss of the actually produced outputs when computing noisy expected risk minimization (ERM) models. They consider the setting where a model does not need to fulfill a predetermined privacy requirement, but instead its loss should not exceed a predetermined value. They propose algorithms to compute the most private model that still fulfills the loss requirement. The authors introduce \emph{ex\hyp post DP} to measure the privacy of a model, which is the special case of our ODP definition when \(\delta=0\). This is why all algorithms proposed in their paper are compatible with our new composition theorem. As opposed to us, Ligett et al.\ do not provide a way to go from ex\hyp post DP to standard DP. Our ODP framework thus widens the applicability of their mechanisms.

We are not the first to employ automated reasoning to
verify differential privacy.
While we restrict ourselves to verifying our abstract composition
theorem, others have gone a step further
and developed tools to verify differential privacy of
concrete programs.
Barthe et al.\ \cite{barthe2013probabilistic} developed a specialized Hoare logic, later extended by Barthe and other
colleagues \cite{barthe2016proving},
and implemented this logic in the toolbox CertiPriv, based on
the Coq proof assistant \cite{bertot2013interactive}.
An alternative approach by Barthe et al.\ \cite{barthe2014proving} transforms probabilistic programs into nonprobabilistic programs such that proving the transformed program to fulfill a certain specification establishes differential privacy of the original program. 
Later approaches \cite{zhang2017lightdp,wang2019proving,zhang2020testing,wang2020checkdp,bichsel2018dpfinder} rely on the SMT solver Z3 \cite{demoura2008z3},
the MaxSMT solver $\nu$Z \cite{bjorner2015nuZ}, or the probabilistic analysis tool PSI \cite{gehr2016psi}
to minimize the manual effort necessary to prove or disprove differential privacy.
Wang et al.'s tool DPGen \cite{wang2021dpgen} can even transform programs violating differential privacy into
differentially private ones.
Recent work by Bichsel et al.\ \cite{bichsel2021dpsniper} uses machine learning to detect differential privacy violations.

\section{Output differential privacy}
\label{sec:method}
In this section we introduce our ODP framework. It is an extension of DP: it contains DP as a special case, but allows for more precise, output\hyp specific privacy accounting.

\begin{definition}[Output Differential Privacy (ODP)]
\label{def:odp}
Let \(\mathcal{O}\) be a set, and let \(\mathcal{P}=\{P_k\}_{k\in\mathcal{K}}\) be a partition of \(\mathcal{O}\), where \(\mathcal{K}\) is a countable index set. Let \(\mathcal{E}:\mathcal{P}\rightarrow \R_{\geq 0}\) be a function that assigns to each set in the partition a non\hyp negative value, and let \(\delta \geq 0\). A randomized mechanism \(M\) with output set \(\mathcal{O}\) is called \emph{\((\mathcal{P},\mathcal{E},\delta)\)\hyp output differentially private} (\((\mathcal{P},\mathcal{E},\delta)\)\hyp ODP) if for all \(S\subseteq \mathcal{O}\) and for all neighboring databases \(\x,\xp\):
\begin{equation*}
    \Pr(M(\x) \in S) \leq \delta + \sum_{k\in\mathcal{K}} e^{\mathcal{E}(P_k)} \Pr(M(\xp) \in S\cap P_k),
\end{equation*}
where the probability space is over the coin flips of the mechanism \(M\).
If \(M\) is \((\mathcal{P},\mathcal{E},\delta)\)\hyp output differentially private for some \(\mathcal{E}\) and \(\delta\), we call \(\mathcal{P}\)
an \emph{ODP partition} for \(M\).
\end{definition}

As an example, consider the mechanism \(M_{\text{toy}}\) from the introduction. An ODP partition for \(M_{\text{toy}}\) would be \(P_1 = \R,\ P_2 = \{\bot\}\), with \(\mathcal{E}(P_1)=\varepsilon,\ \mathcal{E}(P_2)=0\) and \(\delta=0\).

Note that the assumption of the countability of \(\mathcal{P}\) is a technical one that is required in the proof of our composition theorem (\Thmref{thm:odpcomp}), but not a restriction in practice due to the finiteness (and thus countability) of computer representations.

The following results allow us to convert a DP guarantee to an ODP guarantee (\Lemmaref{thm:dp_to_odp}) and an ODP guarantee to a DP guarantee (\Lemmaref{thm:odp_to_dp}):

\begin{lemma}
\label{thm:dp_to_odp}
Let \(M\) be an \((\varepsilon,\delta)\)\hyp differentially private mechanism. Then \(M\) is \((\mathcal{P},\mathcal{E},\delta)\)\hyp output differentially private for any partition \(\mathcal{P}\) of \(\Range(M)\) and the constant function \(\mathcal{E}\equiv\varepsilon\).
\end{lemma}
\begin{proof}
Follows directly from the definitions of DP and ODP.
\end{proof}

\begin{lemma}
\label{thm:odp_to_dp}
Let \(M\) be a \((\mathcal{P},\mathcal{E},\delta)\)\hyp output differentially private mechanism. Then \(M\) is \((\sup_{P\in\mathcal{P}} \mathcal{E}(P), \delta)\)\hyp differentially private.
\end{lemma}
\begin{proof}
Let \(\varepsilon^*=\sup_{P\in\mathcal{P}} \mathcal{E}(P)\). Let \(\x,\xp\) be neighboring databases and let \(S\subseteq\Range(M)\). Then
\begin{align*}
    \Pr(M(\x)\in S) &\leq \delta + \sum_{k\in\mathcal{K}} e^{\mathcal{E}(P_k)} \Pr(M(\xp) \in S\cap P_k)\\
    &\leq \delta + \sum_{k\in\mathcal{K}} e^{\varepsilon^*} \Pr(M(\xp) \in S\cap P_k)\\
    &= \delta + e^{\varepsilon^*} \Pr(M(\xp) \in S).
\end{align*}
\end{proof}

We sometimes want to build up an ODP guarantee from privacy guarantees that only hold w.r.t.\ subsets of the output set of a mechanism. We call such guarantees \textit{subset differential privacy} guarantees, and show how they can be combined into an ODP guarantee (\Lemmaref{thm:sum_deltas}). However, as we show in \Secref{sec:comp_direct}, this does not always result in an optimal ODP guarantee.

\begin{definition}[Subset Differential Privacy]
Let \(\mathcal{O}\) be a set and \(R\subseteq\mathcal{O}\) a subset of \(\mathcal{O}\). Let \(\varepsilon\geq 0\) and \(\delta\geq 0\). A randomized mechanism \(M\) with output set \(\mathcal{O}\) is called \emph{\((R,\varepsilon,\delta)\)\hyp subset differentially private} if for all \(S\subseteq R\) and for all neighboring databases \(\x,\xp\):
\begin{equation*}
    \Pr(M(\x) \in S) \leq e^{\varepsilon} \Pr(M(\xp) \in S) + \delta,
\end{equation*}
where the probability space is over the coin flips of the mechanism \(M\).
\end{definition}

\begin{lemma}
\label{thm:sum_deltas}
Let \(\mathcal{O}\) be a set, and let \(\mathcal{P}=\{P_k\}_{k\in\mathcal{K}}\) be a partition of \(\mathcal{O}\), where \(\mathcal{K}\) is a countable index set. Let \(M\) be a randomized mechanism with output set \(\mathcal{O}\) and let \(\mathcal{E}:\mathcal{P}\mapsto \R_{\geq 0}\) and \(\Delta:\mathcal{P}\mapsto \R_{\geq 0}\) be functions such that \(M\) is \((P_k,\mathcal{E}(P_k),\Delta(P_k))\)\hyp subset differentially private for all \(k\in\mathcal{K}\). Then \(M\) is \((\mathcal{P},\mathcal{E},\sum_{k\in\mathcal{K}}\Delta(P_k))\)\hyp output differentially private.
\end{lemma}
\begin{proof}
Let \(S\subseteq\mathcal{O}\) and let \(\x,\xp\) be neighboring databases. Then
\begin{align*}
    &\Pr(M(\x) \in S) = \sum_{k\in\mathcal{K}} \Pr(M(\x) \in S\cap P_k)\\
    &\leq \sum_{k\in\mathcal{K}} \left[e^{\mathcal{E}(P_k)} \Pr(M(\xp) \in S\cap P_k) + \Delta(P_k)\right]\\
    &= \sum_{k\in\mathcal{K}}\Delta(P_k) + \sum_{k\in\mathcal{K}} e^{\mathcal{E}(P_k)} \Pr(M(\xp) \in S\cap P_k).
\end{align*}
\end{proof}

\subsection{Composition of ODP mechanisms}
As mentioned in the introduction, ODP can be used to give better utility when composing mechanisms. As in classical adaptive composition \cite{dwork2010boosting}, we model composition as a game between a data curator (in the real world this would be the entity with access to the private databases) and an adversary (in the real world a data analyst) that is allowed to spend a total privacy budget of \((\varepsilon_t, \delta_t)\) (in terms of DP) on mechanism invocations (\Algref{alg:type1_comp}). In each round, the adversary chooses a mechanism \(M_i\) and a pair of neighboring databases \(\x^{i,0},\x^{i,1}\). Based on a bit \(b\) that is only known to the data curator, the data curator returns \(M_i(\x^{i,b})\). The adversary may base their choices in an iteration on the mechanism outputs it has seen in previous iterations. The goal of the adversary is to infer \(b\) from the mechanism outputs. Our goal is to ensure that this is not possible with high confidence, by bounding how much the output distribution under \(b=0\) can differ from the output distribution under \(b=1\).
Note that this is a hypothetical game that is required for the privacy analysis. In the real world there is no bit \(b\), and only one private database in each round, on which the mechanism in that round is invoked.

As opposed to previous composition experiments, we require the adversary to not only choose each mechanism \(M_i\) such that its DP guarantee (as computed via \Lemmaref{thm:odp_to_dp}) does not exceed the remaining privacy budget, but to also return an ODP partition \(\mathcal{P}_i\) for \(M_i\). Due to \Lemmaref{thm:dp_to_odp}, each DP mechanism has a trivial associated ODP partition, thus this requirement does not exclude any DP mechanisms. However, if the ODP partition is non-trivial, such as for the mechanisms in \Secref{sec:analysis_existing}, \ref{sec:iteration_length} and \ref{sec:recovering_budget}, the partition can be used to save privacy budget: let \(k\in\mathcal{K}_i\) be such that the output of \(M_i\) falls into the set \(P_{i,k}\) of \(M_i\)'s ODP partition. If \(\mathcal{E}_i(P_{i,k})<\sup_{P\in\mathcal{P}_i}\mathcal{E}_i(P)\), then the incurred \(\varepsilon\)\hyp cost is smaller than the \(\varepsilon\) of the mechanism's DP guarantee.

Something that sets \Algref{alg:type1_comp} apart from the classic composition experiment, but has been used in combination with the privacy filters and odometers of Rogers et al.\ \cite{rogers2016privacy}, is that the privacy guarantees of the mechanisms do not have to be fixed ahead of time. Instead, the adversary can adaptively, i.e., based on the outputs of previous mechanisms, choose the privacy parameters of the next mechanism that they want to invoke. This also means that the adversary can adaptively choose the number of iterations: if they want to spend all of the privacy budget in the first \(I'<I\) iterations, they can choose a mechanism that always produces the same output independently of the database and thus is \((0,0)\)\hyp DP for the remaining \(I-I'\) iterations.
Like Rogers et al., we limit the maximal number of iterations by a fixed number \(I\). This is purely for technical reasons and not a limitation in practice, since \(I\) can be chosen arbitrarily large.
Note that our composition experiment can almost equivalently be formulated using a new variant of privacy filters (see \Appref{sec:adv_comp}). We choose the formulation in \Algref{alg:type1_comp} throughout most of the paper for a more easily accessible presentation.

We show that our composition scheme provides DP:
\makeatletter
\algnewcommand{\StateCont}[1]{\Statex \hskip\ALG@thistlm #1}
\makeatother
\begin{algorithm}[tb]
\caption{\(\odpc(\mathcal{A},\varepsilon_t,\delta_t,b)\)}
\label{alg:type1_comp}
\begin{algorithmic}[1]
\State {Select coin tosses \(R^b_\mathcal{A}\) for \(\mathcal{A}\) uniformly at random.}
\State{Set remaining budget \(\varepsilon_r=\varepsilon_t,\ \delta_r=\delta_t\).}
\For{\(i=1,\dots,I\)}
    \State{\(\mathcal{A} = \mathcal{A}(R^b_\mathcal{A},\{A^b_j\}_{j=1}^i)\) chooses}
    \StateCont{~\(\bullet\)~ neighboring databases \(\x^{i,0},\x^{i,1}\),}
    \StateCont{~\(\bullet\)~ a triple \((\mathcal{P}_i=\{P_{i,k}\}_{k\in\mathcal{K}_i},\mathcal{E}_i,\delta_i)\), and}
    \StateCont{~\(\bullet\)~ a mechanism \(M_i\)}
    \StateCont{such that}
    \StateCont{~\(\bullet\)~ \(\sup_{P\in\mathcal{P}_i}\mathcal{E}_i(P)\leq\varepsilon_r\),} \StateCont{~\(\bullet\)~ \(\delta_i\leq\delta_r\), and }
    \StateCont{~\(\bullet\)~ \(M_i\) is \((\mathcal{P}_i,\mathcal{E}_i,\delta_i)\)\hyp ODP}
    \State{Sample \(A^b_i = M_i(\x^{i,b})\)}
    \State{Let \(k\) be such that \(A^b_i\in P_{i,k}\)}\label{algline:odp1}
    \State{\(\varepsilon_r \minuseq \mathcal{E}_i(P_{i,k})\)}\label{algline:odp2}
    \State{\(\delta_r \minuseq \delta_i\)}\label{algline:odp3}
    \State{\(\mathcal{A}\) receives \(A^b_i\)}
\EndFor
\State{\Return{view \(V^b = (R^b_\mathcal{A}, A^b_1, \dots, A^b_I)\)}}
\end{algorithmic}
\end{algorithm}

\begin{theorem}
\label{thm:odpcomp}
For every adversary \(\mathcal{A}\) and for every set of views \(\mathcal{V}\) of \(\mathcal{A}\) returned by \Algref{alg:type1_comp} we have that
\begin{equation*}
    \Pr(V^0\in\mathcal{V}) \leq e^{\varepsilon_t} \Pr(V^1\in\mathcal{V}) + \delta_t.
\end{equation*}
\end{theorem}

We defer the proof to \Appref{sec:proof_main_thm}, where we first show the theorem statement for a composition length of \(I=2\). We define sets of views \(\mathcal{V}^k\) where the outputs of \(M_1\) come from the same set \(P_k\) of \(M_1\)'s ODP partition. Then we apply an extension of a proof of the simple composition theorem for \((\varepsilon,\delta)\)\hyp mechanisms by Dwork and Lei \cite[Lemma~28]{dwork2009differential} to such sets \(\mathcal{V}_k\). By taking unions over \(\mathcal{V}_k\)'s, we can analyze arbitrary sets of views. For this we make use of the countability of ODP partitions. The general theorem statement finally follows by induction.

\subsection{Formal verification of the composition theorem}

We have formally verified the proof of \Thmref{thm:odpcomp} in the proof assistant
Lean \cite{demoura2015lean}.
The formal proof is available online\footnote{\url{https://doi.org/10.6084/m9.figshare.19330649}}.
The effort of formalizing our proof has paid off:
During the process, we
discovered an error in a previous version of \Thmref{thm:odpcomp}
that all authors and reviewers had previously missed.
At first glance, in \Algref{alg:type1_comp}, it might seem as if the $\delta_i$
that is subtracted from $\delta_r$ on line~\ref{algline:odp3}
could also depend on $P_{i,k}$.
A counterexample shows that this is a fallacy, 
which was subtly hidden in a previous version of our proof.

Unlike our pen-and-paper proof in \Appref{sec:proof_main_thm}, the mechanized version
discusses all questions of measurability. The \texttt{measurability} tactic
of Lean's mathematical library
could resolve many measurability proofs automatically, but some of them had to be
carried out manually. For example, showing that \Algref{alg:type1_comp} is measurable (as 
a function from the sample space to the resulting view)
requires an induction over the number of iterations, which is out of reach of automation.
Apart from this hurdle, Lean's mathematical library \cite{mathlib2020} was surprisingly
mature for our purposes, given that it is still relatively young.

\section{ODP analysis of existing mechanisms}
\label{sec:analysis_existing}
In this section, we analyze two well\hyp known DP mechanisms using our ODP framework.

\subsection{Sparse Vector Technique}

\begin{algorithm}[tb]
\caption{\(M_{\text{SVT}}(\x,Q,\Delta,T_1,T_2,\dots,c)\)}
\label{alg:svt}
\begin{algorithmic}[1]
\State{\(\rho \sim \Lap(\Delta/\varepsilon_1)\)}
\State{\(\text{count}=0\)}
\For{\(q_i\in Q\)}
    \If{\(q_i=\text{STOP}\) \textbf{or} \(\text{count}\geq c\)}
        \State{BREAK}
    \EndIf
    \State{\(\nu_i \sim \Lap(2c\Delta/\varepsilon_2)\)}
    \If{\(q_i(\x) + \nu_i \geq T_i + \rho\)}
        \State{Output \(a_i = \top\)}
        \State{\(\text{count} \pluseq 1\)}
    \Else
        \State{Output \(a_i = \bot\)}
    \EndIf
\EndFor
\end{algorithmic}
\end{algorithm}

The Sparse Vector Technique (SVT) \cite{dwork2009complexity,svt2017} is a method for releasing the results of a sequence of threshold comparisons with DP. There are multiple variants of SVT. We work with the improved variant of the mechanism due to Lyu et al.\ \cite[Alg.~1]{svt2017}; see Alg.~\ref{alg:svt} (\(M_{\text{SVT}}\)). In this variant, the data analyst sends a stream \(Q\) of adaptively chosen \(\R\)\hyp valued queries \(q_i\) and thresholds \(T_i\) to the data curator, who adds the same noise \(\rho\) to the thresholds, and different noise values \(\nu_i\) to the query results. For each query \(i\) they then return the result of the comparison \(q_i(\x) + \nu_i \geq T_i + \rho\): if the inequality holds, \(\top\) is returned, otherwise \(\bot\) is returned. The data curator then moves on to the next query until the stream ends --- which, as opposed to Lyu et al., we make explicit by letting the data analyst send a \(\text{STOP}\) query ---, or until a prespecified number \(c\) of \(\top\)\hyp outputs has been produced. What sets the SVT apart from other DP mechanisms is that for a fixed privacy budget an arbitrary number of \(\bot\) outputs can be produced; however, only a limited number of \(\top\) outputs. Lyu et al.\ show that \(M_{\text{SVT}}\) fulfills \((\varepsilon_1 + \varepsilon_2)\)\hyp DP. From their proof it can be seen that all \(\bot\)\hyp outputs together contribute \(\varepsilon_1\) to the privacy guarantee and each of the at most \(c\) \(\top\)\hyp outputs contributes \(\varepsilon_2/c\). Thus, intuitively, we should be able to save privacy budget if less than \(c\) \(\top\)\hyp outputs are produced. By slightly modifying the proof by Lyu et al., we can show the following lemma:
\begin{lemma}
\label{thm:svt}
For each integer \(c'\in [0,c]\), let \(S^{\text{SVT}}_{c'}\) be the set of outputs of \(M_{\text{SVT}}\) with \(c'\) \(\top\)\hyp entries. Let further \(\mathcal{P}=\{S^{\text{SVT}}_{c'}\}_{c'=0}^c\) and let, for \(0\leq c'\leq c\),
\begin{equation*}
    \mathcal{E}(S^{\text{SVT}}_{c'}) = \varepsilon_1 + \frac{c'}{c} \varepsilon_2
\end{equation*}
Then \(M_{\text{SVT}}\) is \((\mathcal{P},\mathcal{E},0)\)\hyp ODP.
\end{lemma}
\begin{proof}
Our proof follows closely the one of Lyu et al.\ \cite[Thm.~1]{svt2017}. The only differences are that we explicitly consider the query at which the stream \(Q\) stops and that we do not bound \(c'\) by \(c\), but work with the exact value.
Let \(\x, \xp\) be neighboring databases and assume that the sensitivity of all queries \(q_i\) is bounded by \(\Delta\). Let
\begin{align*}
    f_i(\x,z) &= \Pr(q_i(\x) + \nu_i < T_i + z),\\
    g_i(\x,z) &= \Pr(q_i(\x) + \nu_i \geq T_i + z).
\end{align*}
We have that
\begin{align}
    g_i(\x,z-\Delta) &=\Pr(q_i(\x) + \nu_i\geq T_i + z - \Delta)\nonumber\\
    &\leq \Pr(q_i(\xp) + \Delta + \nu_i\geq T_i + z - \Delta)\nonumber\\
    &= \Pr(q_i(\xp) + \nu_i\geq T_i + z - 2\Delta)\nonumber\\
    &\leq e^{\varepsilon_2/c} \Pr(q_i(\xp) + \nu_i\geq T_i + z).\label{eq:bound_gi}
\end{align}
Let \(l>0,\ c'\leq c\) and let \(a\in\{\bot,\top\}^l\cap S^{\text{SVT}}_{c'}\) be an output of \(M_{\text{SVT}}\) of length \(l\) that contains \(c'\) \(\top\)'s. Let \(I_{\bot}=\{i\mid a_i=\bot\},\ I_{\top}=\{i\mid a_i=\top\}\). Then
\begin{equation}
    \label{eq:number_of_tops}
    \abs{I_{\top}}=c'.
\end{equation}
By \(\Pr(q_{l+1}=\text{ST}\mid a)\) denote the probability that the data analyst chooses the STOP query as the next query after having received the \(l\) outputs in \(a\), and let \(z_-=z-\Delta\). We have:

\begin{align*}
    &\frac{\Pr(M_{\text{SVT}}(\x) = a)}{\Pr(M_{\text{SVT}}(\xp) = a)}\\
    &= \frac{\Pr(M_{\text{SVT}}(\x) = a\mid q_{l+1}=\text{ST})\Pr(q_{l+1}=\text{ST}\mid a)}{\Pr(M_{\text{SVT}}(\xp) = a\mid q_{l+1}=\text{ST})\Pr(q_{l+1}=\text{ST}\mid a)}\\
    &= \frac{\Pr(M_{\text{SVT}}(\x) = a\mid q_{l+1}=\text{ST})}{\Pr(M_{\text{SVT}}(\xp) = a\mid q_{l+1}=\text{ST})}\\
    &= \frac{\int\limits_{-\infty}^{\infty}\Pr(\rho=z) \prod\limits_{i\in I_{\bot}} f_i(\x,z) \prod\limits_{i\in I_{\top}} g_i(\x,z) dz}{\int\limits_{-\infty}^{\infty}\Pr(\rho=z) \prod\limits_{i\in I_{\bot}} f_i(\xp,z) \prod\limits_{i\in I_{\top}} g_i(\xp,z) dz}\\
    &= \frac{\int\limits_{-\infty}^{\infty}\Pr(\rho=z_-) \prod\limits_{i\in I_{\bot}} f_i(\x,z_-) \prod\limits_{i\in I_{\top}} g_i(\x,z_-) dz}{\int\limits_{-\infty}^{\infty}\Pr(\rho=z) \prod\limits_{i\in I_{\bot}} f_i(\xp,z) \prod\limits_{i\in I_{\top}} g_i(\xp,z) dz}\\
    &\leq \frac{\int\limits_{-\infty}^{\infty}e^{\varepsilon_1}\Pr(\rho=z) \prod\limits_{i\in I_{\bot}} f_i(\xp,z) \prod\limits_{i\in I_{\top}} g_i(\x,z_-) dz}{\int\limits_{-\infty}^{\infty}\Pr(\rho=z) \prod\limits_{i\in I_{\bot}} f_i(\xp,z) \prod\limits_{i\in I_{\top}} g_i(\xp,z) dz}\\
    &\overset{\text{(\ref{eq:bound_gi})}}{\leq} \frac{\int\limits_{-\infty}^{\infty}e^{\varepsilon_1}\Pr(\rho=z) \prod\limits_{i\in I_{\bot}} f_i(\x,z) \prod\limits_{i\in I_{\top}} e^{\varepsilon_2/c}g_i(\x,z) dz}{\int\limits_{-\infty}^{\infty}\Pr(\rho=z) \prod\limits_{i\in I_{\bot}} f_i(\xp,z) \prod\limits_{i\in I_{\top}} g_i(\xp,z) dz}\\
    &\overset{\text{(\ref{eq:number_of_tops})}}{=} e^{\varepsilon_1}(e^{\varepsilon_2/c})^{c'} = e^{\varepsilon_1 + \frac{c'}{c} \varepsilon_2}.
\end{align*}
Since this bound hold for every element \(a\in S^{\text{SVT}}_{c'}\), it also holds for all subsets of \(S^{\text{SVT}}_{c'}\). Hence, \(M_{\text{SVT}}\) is \((S^{\text{SVT}}_{c'},e^{\varepsilon_1 + \frac{c'}{c} \varepsilon_2},0)\)\hyp subset differentially private for every \(c'\leq c\). The ODP bound then follows from \Lemmaref{thm:sum_deltas}.
\end{proof}

\xhdr{Application}
A common use case for SVT is the differentially private release of only those entries of a vector with large magnitude, instead of the entire vector \cite{li2019privacy,zhang2021wide}. This can be desirable for multiple reasons: to be able to release the entries with less noise, since the privacy budget needs to be divided between fewer entries; to release only those values of a histogram that are large enough in magnitude so that they will not be dominated by the added noise; or to reduce communication costs in a distributed setting. Assume that \(f\) is a function that takes as input the private database and returns the vector of interest. The queries would then be \(q_i(\x) = \abs{f_i(\x)}\), i.e., the absolute value of the \(i\)\hyp th entry of the vector \(f(\x)\). Only if the corresponding output \(a_i\) is \(\top\), a differentially private version of \(f_i(\x)\) is released.
When analyzing SVT with DP \cite[Sec.~4.1]{svt2017}, the available privacy budget \(\varepsilon_t=\varepsilon_1+\varepsilon_2+\varepsilon_3\) is divided into a budget \(\varepsilon_1+\varepsilon_2\) for SVT itself, and a budget \(\varepsilon_3\) for the release of the vector entries. When using the Laplace mechanisms with the same variance to perturb the vector entries, a privacy budget of \(\varepsilon_3/c\) is available for each of the at most \(c\) entries that get released. However, it is not guaranteed that \(c\) entries will surpass the threshold and thus get released. ODP allows us to do the following: first invoke SVT. Let \(c'\leq c\) be the number of entries that surpassed the threshold. We have now used up a privacy budget of \(\varepsilon_1 + \frac{c'}{c}\varepsilon_2\) and thus have a privacy budget of \(\varepsilon'_3=\varepsilon_3 + \frac{c-c'}{c}\varepsilon_2\) available for the release of the \(c'\) vector entries. Assume that \(c'<c\). Then we have \(\varepsilon'_3>\varepsilon_3\). Further, we only need to split this budget between \(c'\) instead of between \(c\) vector entries and thus have a budget of \(\varepsilon'_3/c'\) per entry. This is not possible with a pure DP analysis, where we always have to assume the worst\hyp case of \(c\) released vector entries. These two sources of budget saving --- from SVT itself via \Lemmaref{thm:svt}, and from making use of the knowledge of the actual number of released entries --- lead to less noise in the released entries. Adapting the noise per entry to the actual number of released entries would also be possible with privacy filters \cite{rogers2016privacy}, but saving budget from SVT itself is only possible with ODP.

In \Figref{fig:svt} we numerically demonstrate the advantage of using ODP when releasing a sparse vector via SVT\footnote{The code for generating the plots in this paper can be found at \url{https://doi.org/10.6084/m9.figshare.19330649}.}. The vector in our example has \(100\) entries that each have a value of either \(0\) or \(1000\). We assume that at most \(20\) of the entries lie above the threshold \(T=500\), i.e., \(c=20\). The total privacy budget is \(\varepsilon=1\), split into a budget \(\varepsilon_1+\varepsilon_2=1/2\) for determining the indices of the large vector entries and a budget \(\varepsilon_3=1/2\) for releasing the corresponding values. For splitting budget between \(\varepsilon_1\) and \(\varepsilon_2\) we choose the optimal ratio \(\varepsilon_1/\varepsilon_2 = 1/(2c)^{2/3}\) according to Lyu et al.\ \cite{svt2017}. We further assume that the queries \(q_i(\x) = \abs{f_i(\x)}\) have sensitivity \(\Delta=1\).
Since \(c\) is only an upper bound on the number of large entries, the actual number of large entries can be lower. The ODP analysis can exploit cases with less than \(c\) large entries for adding less noise to the values of the released entries, while with the standard analysis of SVT always the same amount of noise has to be added.  For different numbers of entries with value \(1000\) we compare the amount of noise added to the released vector entries with and without an ODP analysis. Without an ODP analysis this is a fixed amount, whereas with an ODP analysis the amount of noise depends on the number of released vector entries, which itself depends on the outcome of a random comparison. \Figref{fig:svt} therefore shows an upper bound on the expected amount of noise for the ODP analysis based on Chebyshev's inequality, where the expectation is taken over the randomness of \(\rho\) and the \(\nu_i\) in \Algref{alg:svt}.

\subsection{Propose-Test-Release}

Propose-Test-Release (PTR) by Dwork and Lei \cite{dwork2009differential} is a framework for designing differentially private mechanisms based on the idea of connecting DP and robust statistics. Using PTR, the authors derive mechanisms for differentially privately estimating the median, the trimmed mean, the interquartile range and the coefficients of a linear regression model. The structure of PTR mechanisms is that they first propose a bound on the local sensitivity of the query function and then differentially privately test whether the local sensitivity lies below this bound. If the local sensitivity lies below the proposed bound, the result of the query with noise adapted to the bound is returned. Otherwise \(\bot\) (no response) is returned. In the case of a \(\bot\)\hyp output, the local sensitivity was high, meaning that the query result likely would have not been robust to changes in individual data points, and thus not very useful anyway. The basic building blocks of mechanisms based on PTR are \emph{\((\varepsilon,\delta)\)\hyp PTR functions}, or short, PTR functions. They get as input a database \(\x\) and a value \(s\) that may be used for computing a proposed bound on the local sensitivity. For example, Dwork and Lei propose a PTR function for computing the \(\alpha\)\hyp trimmed mean, i.e., the mean of an empirical distribution when discarding the upper and lower \(\alpha/2\)\hyp quantile. The local sensitivity, that is, the sensitivity w.r.t.\  a concrete database, of the \(\alpha\)\hyp trimmed mean depends on the distance between the \(\alpha/2\)- and the \((1-\alpha/2)\)\hyp quantile. The authors provide a mechanism for computing this distance (a generalization of the mechanism \(M_{\text{IQR}}\) below), and use the output of this mechanism as the second input \(s\) to the PTR function for the \(\alpha\)\hyp trimmed mean.
Formally, PTR functions are defined as follows:

\begin{figure}[t!]
    \centering
    \includegraphics[width=0.48\textwidth]{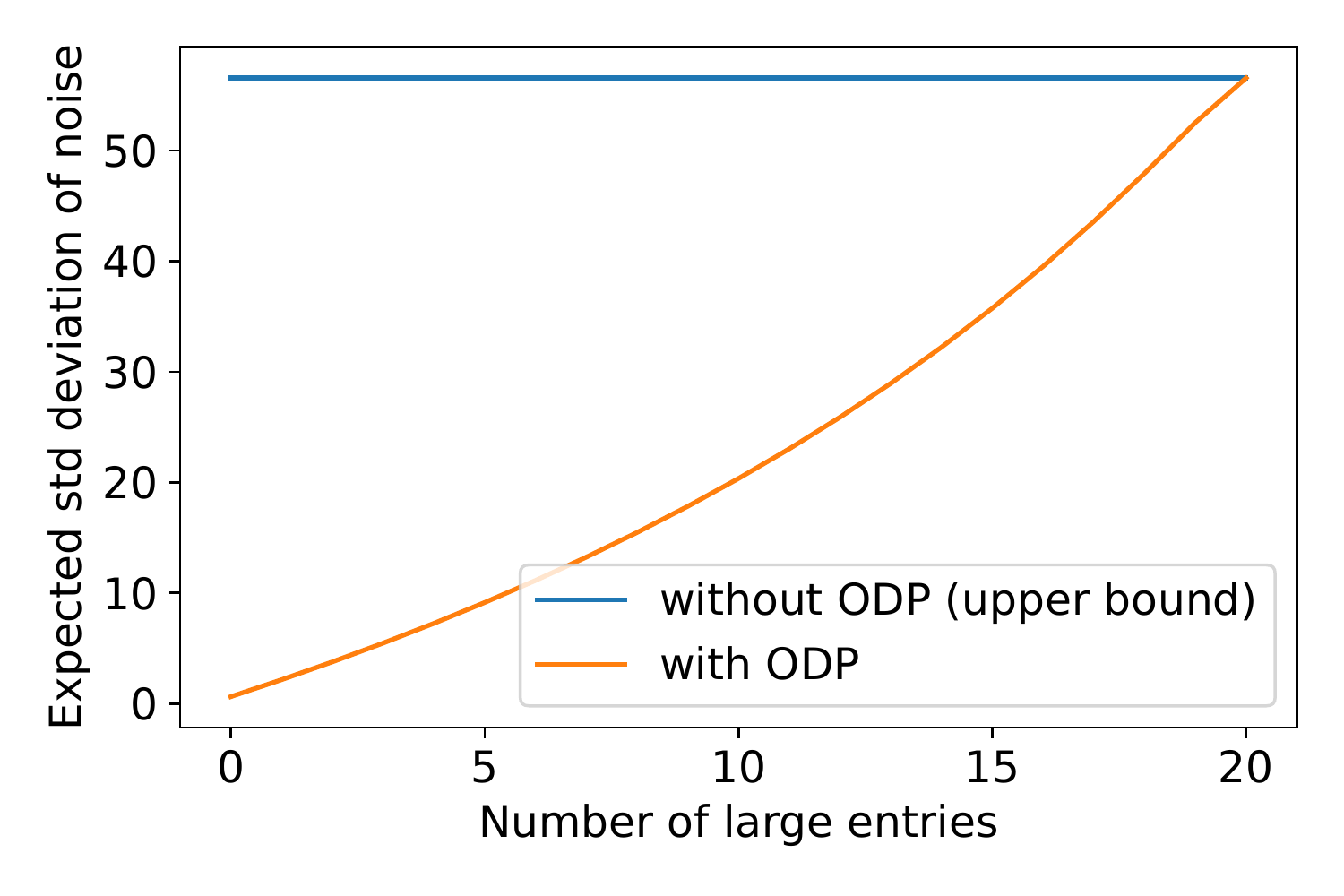}
    \caption{Noise required with and without an ODP analysis when releasing the values of at most \(20\) large entries of a 100\hyp dimensional sparse vector with a total privacy budget of \(\varepsilon=1\). We vary the number of large entries in the vector. Large entries have value \(1000\), small entries have value \(0\), and the threshold for releasing an entry is \(500\).}
    \label{fig:svt}
\end{figure}

\begin{definition}[\((\varepsilon,\delta)\)\hyp PTR function]
A function \(T(\x,s):\mathcal{D}\times(\mathcal{C}\cup\{\bot\})\rightarrow\mathcal{C}'\cup\{\bot\}\) is \emph{\((\varepsilon,\delta)\)\hyp PTR} if
\begin{enumerate}
    \item \(\Pr(T(\x,\bot)=\bot) = 1\) for all \(\x\in\mathcal{D}\).
    \item For all \(s\in \mathcal{C}\), \(\x,\xp\in\mathcal{D}\) neighboring,
    \begin{align*}
        \Pr(T(\x,s)=\bot) &\leq e^{\varepsilon} \Pr(T(\xp,s)=\bot),\\
        \Pr(T(\x,s)\neq\bot) &\leq e^{\varepsilon} \Pr(T(\xp,s)\neq\bot).
    \end{align*}\label{enum:ptr_eps}
    \item There exists \(G(T,\x)\subseteq \mathcal{C}\) such that if \(s\in G(T,\x)\), then for all \(\xp\) neighboring \(\x\) and all \(C'\subseteq\mathcal{C}'\),
    \begin{equation*}
        \Pr(T(\x,s)\in C') \leq e^{2\varepsilon} \Pr(T(\xp,s)\in C');
    \end{equation*}
    \item for all \(\x\in\mathcal{D}\), if \(s\notin G(T,\x)\): \(\Pr(T(\x,s)\neq\bot)\leq\delta\).
\end{enumerate}
\end{definition}
Let \(T\) be a PTR function. If \(s=\bot\) and thus no bound on the local sensitivity can be computed from \(s\), \(T\) returns \(\bot\). The conditions in \ref{enum:ptr_eps} allow for improved privacy bounds when composing PTR functions where the next function is only invoked if the previous one did not return \(\bot\) (see the example of \(M_{\text{IQR}}\) below). In that case, all PTR function invocations except from the last one only increase the privacy budget spending by \(\varepsilon\) instead of \(2\varepsilon\). Not every proposed sensitivity bound is large enough for every dataset to ensure privacy when adding noise according to this bound. The bounds that are admissible for a dataset \(\x\) are described by the set \(G(T,\x)\). If a proposed bound is too small, then with probability at least \(1-\delta\) the \(\bot\)\hyp symbol will be returned to prevent revealing too much information about \(\x\).

As can easily be seen, a PTR function \(T\) fulfills \((2\varepsilon,\delta)\)\hyp DP. Letting \(P_1=\mathcal{C}',\ P_2=\{\bot\}\) and \(\mathcal{E}(P_1)=2\varepsilon,\ \mathcal{E}(P_2)=\varepsilon\), \(T\) also fulfills \((\{P_1,P_2\},\mathcal{E},\delta)\)\hyp ODP. Hence we can save privacy budget in the case of a \(\bot\)\hyp output. We exemplify this for the mechanism proposed by Dwork and Lei for approximating the interquartile range (IQR) of an empirical distribution on \(\R\), that is, the difference between the 75th and the 25th percentile, which serves as a measure of scale. We use \(M_{\text{IQR}}\) to denote this mechanism.
\(M_{\text{IQR}}\) works by discretizing \(\R\) into buckets and proposing a local sensitivity of the discretized IQR. If the IQR has at most the proposed sensitivity, a noisy version of the IQR is released, otherwise \(\bot\) is released. To avoid an unlucky choice of the discretization, \(M_{\text{IQR}}\) uses two discretizations, where the second one is a shifted version of the first one. If the output produced by the mechanism resulting from the first dicretization (denoted by \(M^{\text{D}}_1\)) is not \(\bot\), this output is returned and the computation ends. Otherwise the mechanism resulting from the second discretization (denoted by \(M^{\text{D}}_2\)) is invoked and its output is returned. The authors show that, for \(j=1,2\), \(M^{\text{D}}_j\) is a PTR function. Without having formalized the concept of ODP yet, they use the fact that a \(\bot\)\hyp output of \(M^{\text{D}}_1\) contains less information about the database than a non\hyp \(\bot\)\hyp output, combined with the fact that \(M^{\text{D}}_2\) is only invoked if \(M^{\text{D}}_1\)'s output is \(\bot\) to show that \(M\) fulfills \((3\varepsilon,\delta)\)\hyp DP (instead of the naive \(((2+2)\varepsilon,(1+1)\delta)\)\hyp DP), because one never has to account for a non\hyp \(\bot\)\hyp output of both \(M^{\text{D}}_1\) and \(M^{\text{D}}_2\). The authors show that the same reasoning can be applied to more general compositions of PTR functions to save privacy budget also in the computation of, e.g., the median or of regression parameters.

\xhdr{ODP analysis: Treating \(M^{\text{D}}_1\) and \(M^{\text{D}}_2\) as separate mechanisms}
By treating each PTR function that makes up one of their mechanisms as a separate mechanism and composing them via Alg.~\ref{alg:type1_comp}, we can save \(\varepsilon\)\hyp budget beyond the improved analysis by Dwork and Lei, but require additional \(\delta\)\hyp budget. In the example of \(M_{\text{IQR}}\) we either only invoke \(M^{\text{D}}_1\) (if the output is in \(\R\)) or invoke both \(M^{\text{D}}_1\) and \(M^{\text{D}}_2\). Since \(M^{\text{D}}_1\) and \(M^{\text{D}}_2\) are PTR functions, they fulfill \((\{\R,\{\bot\}\},\mathcal{E},\delta)\)\hyp ODP with \(\mathcal{E}(\R)=2\varepsilon\) and \(\mathcal{E}(\bot)=\varepsilon\). One of the following three cases will occur:
\begin{enumerate}
    \item \(M^{\text{D}}_1\) returns \(r\) for some \(r\in\R\) (and \(M^{\text{D}}_2\) never gets invoked);\label{enum:svt_case_1}
    \item \(M^{\text{D}}_1\) returns \(\bot\) and \(M^{\text{D}}_2\) returns \(r\) for some \(r\in\R\);\label{enum:svt_case_2}
    \item both \(M^{\text{D}}_1\) and \(M^{\text{D}}_2\) return \(\bot\).\label{enum:svt_case_3}
\end{enumerate}
In case~\ref{enum:svt_case_1} we spend a privacy budget of \((2\varepsilon,\delta)\), in case~\ref{enum:svt_case_2} we spend a budget of \((3\varepsilon,2\delta)\), and in case~\ref{enum:svt_case_3} we spend a budget of \((2\varepsilon,2\delta)\). Compared with the DP analysis, in cases \ref{enum:svt_case_1} \& \ref{enum:svt_case_3} we save a budget of \((\varepsilon,0)\), but we spend an additional budget of \((0,\delta)\) in cases \ref{enum:svt_case_2} \& \ref{enum:svt_case_3}. By choosing the order of the two discretizations used for \(M^{\text{D}}_1\) and \(M^{\text{D}}_2\) uniformly at random, we can ensure that with probability at least \(1/2\) we will be in case~\ref{enum:svt_case_1} if at least one of the two mechanisms returns a value in \(\R\).

\xhdr{ODP analysis: Treating \(M_{\text{IQR}}\) as a single mechanism}
When treating \(M_{\text{IQR}}\) as a single mechanism, we can even strictly improve upon the original analysis in terms of ODP. We have the following lemma, which shows that we can save privacy budget if both discretizations result in a \(\bot\)\hyp output:

\begin{lemma}
\label{thm:ptr_single}
Let \(P_1=\R,\ P_2=\{\bot\}\) and \(\mathcal{E}(P_1)=3\varepsilon,\ \mathcal{E}(P_2)=2\varepsilon\). Then \(M_{\text{IQR}}\) is \((\{P_1,P_2\},\mathcal{E},\delta)\)\hyp ODP.
\end{lemma}
\begin{proof}
Since \(M_{\text{IQR}}\) is \((3\varepsilon,\delta)\)\hyp DP, it is in particular \((\R,3\varepsilon,\delta)\)\hyp subset differentially private. Let \(\x\) and \(\xp\) be neighboring databases. Since (1) the randomnesses of \(M^{\text{D}}_1\) and \(M^{\text{D}}_2\) are independent and (2) \(M^{\text{D}}_1\) and \(M^{\text{D}}_2\) are \((\varepsilon,\delta)\)\hyp PTR functions, it holds that
\begin{align*}
    &\Pr(M(\x)=\bot) = \Pr(M^{\text{D}}_1(\x)=\bot,M^{\text{D}}_2(\x)=\bot)\\
    &\overset{\text{(1)}}{=} \Pr(M^{\text{D}}_1(\x)=\bot)\Pr(M^{\text{D}}_2(\x)=\bot)\\
    &\overset{\text{(2)}}{\leq} e^{\varepsilon}\Pr(M^{\text{D}}_1(\xp)=\bot)e^{\varepsilon}\Pr(M^{\text{D}}_2(\xp)=\bot).
\end{align*}
Thus \(M_{\text{IQR}}\) is \((\{\bot\},2\varepsilon,0)\)\hyp subset differentially private. The statement of the lemma then follows from \Lemmaref{thm:sum_deltas}.
\end{proof}

\xhdr{Application}
There are multiple statistics that can be computed via PTR mechanisms. Typically one is not interested in only a single differentially private statistic of a dataset, but in multiple statistics. E.g., one might want to release a measure of the location of the data and a measure of its spread. For this task one could first invoke \(M_{\text{IQR}}\) to compute the IQR of the data and then another DP mechanism \(M_{\text{median}}\) to compute its median \cite{dwork2009differential}. The total privacy budget \((\varepsilon_t,\delta_t)\) would be divided into a budget \((3\varepsilon_1,\delta_1)\) for \(M_{\text{IQR}}\) and a budget \((\varepsilon_2,\delta_2)\) for \(M_{\text{median}}\), where \(3\varepsilon_1 + \varepsilon_2 = \varepsilon_t\) and \(\delta_1 + \delta_2 = \delta_t\). With a DP analysis, \(M_{\text{IQR}}\) would always use up \((3\varepsilon_1,\delta_1)\) from the budget. However, with ODP composition and when using our second ODP analysis (\Lemmaref{thm:ptr_single}), \(M_{\text{IQR}}\) only uses up \((2\varepsilon_1,\delta_1)\) in the case of a \(\bot\)\hyp output, and thus the budget remaining for \(M_{\text{median}}\) increases to \((\varepsilon_1+\varepsilon_2,\delta_2)\) in that case. The increased budget can be used to reduce the amount of noise added in the computation of the differentially private median, which makes the released estimate more accurate.

\section{ODP for mechanisms with variable numbers of iterations}
\label{sec:iteration_length}
There are many iterative algorithms for which the required number of iterations is not known beforehand. Instead, they are executed until a certain criterion is reached, e.g., the norm of the gradient in an optimization problem falls below a prespecified threshold or the validation loss starts increasing in the training of a machine learning (ML) model. A common technique to make the final output of iterative algorithms differentially private is by making the intermediate values in each iteration differentially private. In stochastic gradient descent, for instance, which is often used for training neural networks, this is typically achieved by adding Gaussian noise to each gradient \cite{abadi2016deep,bu2020deep}. The DP guarantees for the intermediate values are then combined using a composition theorem to get a DP guarantee for the final output. For this it is necessary to know beforehand --- i.e., before executing the iterative algorithm --- how many iterations will be performed. If the total privacy budget is fixed, a larger number of iterations means that less privacy budget can be used on each iteration, whereas with a smaller number of iterations more privacy budget is available for each iteration. If, as in many cases, the optimal number of iterations is not known {\em a priori}, the number will often be either overestimated or underestimated. If the number of iterations is overestimated, privacy budget is wasted on iterations that are not required. If it is underestimated, the iterative algorithm will halt before it has reached an optimal solution, e.g., an ML model might have a larger error than it could have with more iterations.

ODP allows us to escape this dilemma. Consider a data analyst that chooses a number \(k\) of iterations for the iterative algorithm. With standard DP analysis, the algorithm always runs for \(k\) iterations, even if it converges at an earlier iteration \(k'<k\). The privacy budget for the remaining \(k-k'\) iterations is thus wasted. With ODP, however, the algorithm can be stopped at iteration \(k'\) and the budget that was reserved for the remaining \(k-k'\) iterations can be used via ODP composition for other queries on the same database or on other databases that might share individuals who contributed data with the original database. This solves the problem of overestimating the number of iterations. To solve the problem of underestimating the number of iterations, the data analyst can purposely choose a large number for \(k\) that is likely to be an overestimate. This is not problematic anymore, since in the case where it was indeed an overestimate, the algorithm can again halt earlier, and the remaining privacy budget can be used for other tasks.

\begin{algorithm}[!tbh]
\caption{\(\im\big[\{k_i\}_{i=1}^n, \{M_k\}_{k=1}^{k_n}, \{h^{k_i}\}_{i=1}^{n-1}\big](\x)\)}
\label{alg:iterative_mechanism}
\begin{algorithmic}[1]
\For{\(k=1,\dots,k_n\)}
    \State{\(s_k=M_k((s_1,\dots,s_{k-1}),\x)\)}
    \If{\(k\in \{k_i\}_{i=1}^{n-1}\)}
        \State{Let \(i\) such that \(k=k_i\)}
        \If{\(h^{k_i}(s_1,\dots,s_{k_i})=1\)}
            \State{\Return{\((s_1,\dots,s_{k_i})\)}}
        \EndIf
    \EndIf
\EndFor
\State{\Return{\((s_1,\dots,s_{k_n})\)}}
\end{algorithmic}
\end{algorithm}

We consider an iterative mechanism \(M\) as defined in \Algref{alg:iterative_mechanism}.
Let \(k_1<k_2<\dots<k_n\) be numbers of iterations after which \(M\) might stop. For \(k=1,\dots,k_n\), let \(M_k\) be a differentially private mechanism that takes as input the output of previous mechanisms and a database. \(M_k\) is the mechanism executed in the \(k\)\hyp th iteration of \(M\). Define a set of binary functions
\begin{equation*}
    h^{k_i}:\prod_{k=1}^{k_i}\Range(M_k)\mapsto\{0,1\},
\end{equation*}
\(i=1,\dots,n-1,\) that act as stopping criteria. After \(k_i\) iterations, \(M\) evaluates \(h^{k_i}\) on the outputs \(s_1,\dots,s_{k_i}\) produced so far: if \(h^{k_i}(s_1,\dots,s_{k_i})=1\), \(M\) halts and returns \((s_1,\dots,s_{k_i})\); otherwise it continues. If \(h^{k_i}(s_1,\dots,s_{k_i})=0\) for all \(i<n\), \(M\) halts and returns \((s_1,\dots,s_{k_n})\) after iteration \(k_n\). Note that we assume for simplicity that \(h^{k_i}\) can be evaluated with only the information that has already been computed in a differentially private way, i.e., it only requires access to \(\x\) via \(M_1,\dots,M_{k_i}\) (this is, e.g., the case if \(h^{k_i}\) is based on the gradient norm in differentially private SGD or the validation score of an ML model on a public validation set). If this is not the case and the computation of \(h^{k_i}\) requires access to the database \(\x\), we can simply add an additional mechanism \(\tilde{M}_{k_i}\) after \(M_{k_i}\) that computes \(h^{k_i}\) in a differentially private way, and reindex the sequence of mechanisms such that \(M_{k_i}=\tilde{M}_{k_i}\).

For \(i=1,\dots,n\), let
\begin{equation*}
    P_i=\{s \mid s\in\Range(M), \abs{s}=k_i\},
\end{equation*}
where \(\abs{s}\) denotes the length of the vector \(s\), and let \(\mathcal{P}=\{P_i\}_{i=1}^n\). Throughout this section we always mean this partition \(\mathcal{P}\) when referring to the partition of the output space of an iterative mechanism. When we want to make the dependency on an iterative mechanism \(M'\) explicit, we write \(\mathcal{P}^{M'}=\{P_i^{M'}\}_{i=1}^n\). In \Secref{sec:comp_sum_deltas} we show a generic ODP bound for \(M\) with the ODP partition \(\mathcal{P}\) based on \Lemmaref{thm:sum_deltas} that is compatible with any DP composition theorem. In \Secref{sec:comp_direct} we show how this generic bound can be improved upon via a direct derivation that depends on the specific composition setting. Thus, \Secref{sec:comp_direct} also acts as a demonstration of how \Lemmaref{thm:sum_deltas} does not always yield an optimal ODP bound.

The iterative mechanism in \Algref{alg:iterative_mechanism} is defined by a sequence \(\{k_i\}_{i=1}^n\) of potential stopping points, a sequence \(\{M_k\}_{k=1}^{k_n}\) of mechanisms that are invoked in the different iterations, and a sequence \(\{h^{k_i}\}_{i=1}^{n-1}\) of stopping criteria. We denote an iterative mechanism, given by such a set of parameters, via \(\im\big[\{k_i\}_{i=1}^n, \{M_k\}_{k=1}^{k_n}, \{h^{k_i}\}_{i=1}^{n-1}\big]\).

\subsection{ODP bound based on \Lemmaref{thm:sum_deltas}}
\label{sec:comp_sum_deltas}
We derive an ODP bound for iterative mechanisms via \Lemmaref{thm:sum_deltas} and any composition theorem \(\comp\) that is compatible with \(M_1,\dots,M_{k_n}\), i.e., that can be applied to the sequence of mechanisms \(M_1,\dots,M_{k_n}\). This could, e.g., be the optimal composition theorem for adaptive, heterogeneous composition \cite{murtagh2016complexity} (see \Secref{sec:comp_direct}), for the most general class of mechanisms.

\begin{lemma}
\label{thm:odp_partition_iterative}
Let
\begin{equation*}
    M=\im\big[\{k_i\}_{i=1}^n, \{M_k\}_{k=1}^{k_n}, \{h^{k_i}\}_{i=1}^{n-1}\big]
\end{equation*}
be an iterative mechanism.
Let \(\comp\) be a DP composition theorem that is compatible with \(M_1,\dots,M_{k_n}\). \(\comp\) takes as input a sequence of mechanisms and a desired value \(\varepsilon\), and returns a value \(\delta\) such that the sequence fulfills \((\varepsilon,\delta)\)\hyp DP. For \(i=1,\dots,n\), let \(\varepsilon^{k_i}\geq 0\) be a desired \(\varepsilon\)\hyp value, and let
\begin{equation*}
    \delta^{k_i} = \comp(M_1,\dots,M_{k_i};\varepsilon^{k_i})
\end{equation*}
be the \(\delta\) returned by the composition theorem.
Then \(M\) fulfills \((\mathcal{P},\mathcal{E},\delta)\)\hyp ODP with
\begin{equation*}
    \mathcal{E}(P_i)=\varepsilon^{k_i}
\end{equation*}
and
\begin{equation*}
    \delta = \sum_{i=1}^n \delta^{k_i}
\end{equation*}
for all \(i=1,\dots,n\).
\end{lemma}
\begin{proof}
We can assume w.l.o.g.\ that the stopping criteria \(h^{k_i}\) are deterministic. If they are not deterministic, we can let \(M_{k_i}\), in addition to its original output, return a sample from the distribution of \(h^{k_i}\)'s randomness, which \(h^{k_i}\) can then access. This does not change the DP guarantee of \(M_{k_i}\), since the distribution of \(h^{k_i}\)'s randomness is independent of the database.

Let \(i\in\{1,\dots,n\}\). Define \(M^{k_i}=(M_1,\dots,M_{k_i})\).
Because \(h^{k_i}\) is deterministic, it holds for any database \(\x\) and any \(S\subseteq P_i\) that
\begin{equation*}
    \Pr(M(\x)\in S) = \Pr(M^{k_i}(\x)\in S).
\end{equation*}
Since \(M^{k_i}\) is \((\varepsilon^{k_i},\delta^{k_i})\)\hyp differentially private, we have, for any neighboring databases \(\x,\xp\):
\begin{align*}
    \Pr(M(\x)\in S) &= \Pr(M^{k_i}(\x)\in S)\\
    &\leq e^{\varepsilon^{k_i}} \Pr(M^{k_i}(\xp)\in S) + \delta^{k_i}\\
    &= e^{\varepsilon^{k_i}} \Pr(M(\xp)\in S) + \delta^{k_i}.
\end{align*}
Hence \(M\) is \((P_{k_i},\varepsilon^{k_i},\delta^{k_i})\)\hyp subset differentially private. The statement of the lemma then follows from \Lemmaref{thm:sum_deltas}.
\end{proof}

\subsection{ODP bound via direct derivation}
\label{sec:comp_direct}

In this subsection we give an example that shows that we can get better ODP bounds for iterative mechanisms than the ones obtained by applying \Lemmaref{thm:odp_partition_iterative}. This comes at the cost of losing generality, since we cannot plug in any existing DP composition theorem anymore, but have to do a derivation from scratch. We consider the adaptive, heterogeneous composition of arbitrary differentially private mechanisms. Our proof extends the one by Murtagh and Vadhan \cite{murtagh2016complexity} for a composition setting without stopping rules to one with stopping rules. Adaptive, heterogenous composition means that, for \(k=1,2,\dots\), we assume that \(M_k\) fulfills \((\varepsilon_k,\delta_k)\)\hyp DP for some fixed \(\varepsilon_k\) and \(\delta_k\), and that \(M_k\) may depend on the outputs of \(M_1,\dots,M_{k-1}\), but we assume nothing beyond that.

Let
\begin{equation*}
    M=\im\big[\{k_i\}_{i=1}^n, \{M_k\}_{k=1}^{k_n}, \{h^{k_i}\}_{i=1}^{n-1}\big]
\end{equation*}
be an iterative mechanism, based on mechanisms \(\{M_k\}_{k=1}^{k_n}\) as described above.
For a function \(\mathcal{E}:\mathcal{P}\mapsto\R_{\geq 0}\) we define the smallest \(\delta\) such that \(M\) fulfills \((\mathcal{P},\mathcal{E},\delta)\)\hyp ODP as

\begin{align*}
    &\optdelta(\{k_i\}_{i=1}^n, \{M_k\}_{k=1}^{k_n}, \{h^{k_i}\}_{i=1}^{n-1}, \mathcal{E})\\
    &= \inf\{\delta\mid \im\big[\{k_i\}_{i=1}^n, \{M_k\}_{k=1}^{k_n}, \{h^{k_i}\}_{i=1}^{n-1}\big]\\
    &\hphantom{= \inf\{\delta\mid}\qquad \text{is } (\mathcal{P},\mathcal{E},\delta)\text{-ODP}\}.
\end{align*}

Given \(\mathcal{E}\), \(\{k_i\}_{i=1}^n\) and a fixed list of DP parameters \(\{(\varepsilon_k,\delta_k)\}_{k=1}^{k_n}\), we want to find the minimal \(\delta\) such that \(\im\big[\{k_i\}_{i=1}^n, \{M_k\}_{k=1}^{k_n}, \{h^{k_i}\}_{i=1}^{n-1}\big]\) fulfills \((\mathcal{P},\mathcal{E},\delta)\)\hyp ODP for any sequence \(\{M_k\}_{k=1}^{k_n}\) of mechanisms that fulfill \((\varepsilon_k,\delta_k)\)\hyp DP, \(k=1,\dots,k_n\), and any sequence of stopping criteria. We thus define
\begin{align*}
    &\optdelta(\{k_i\}_{i=1}^n, \{(\varepsilon_k,\delta_k)\}_{k=1}^{k_n}, \mathcal{E})\\
    &= \sup_{\substack{\{M_k\}_{k=1}^{k_n}\\ \{h^{k_i}\}_{i=1}^{n-1}}} \{\optdelta(\{k_i\}_{i=1}^n, \{M_k\}_{k=1}^{k_n}, \{h^{k_i}\}_{i=1}^{n-1}, \mathcal{E})\\
    &\hphantom{= \sup_{\substack{\{M_k\}_{k=1}^{k_n}\\ \{h^{k_i}\}_{i=1}^{n-1}}} \{}
    \mid M_k \text{ is } (\varepsilon_k,\delta_k)\text{-DP},\ k=1,\dots,k_n\}.
\end{align*}

The remainder of this subsection is devoted to deriving the expression for \(\optdelta(\{k_i\}_{i=1}^n, \{(\varepsilon_k,\delta_k)\}_{k=1}^{k_n}, \mathcal{E})\) in \Thmref{thm:iterative_opt_delta}. Further, following \Thmref{thm:iterative_opt_delta}, we compare this optimal value with the one that can be obtained via \Lemmaref{thm:odp_partition_iterative}.
All proofs from this subsection are deferred to \Appref{sec:proofs_comp_direct}.

Like Kairouz et al.\ \cite{kairouz2017composition} and Murtagh and Vadhan \cite{murtagh2016complexity}, we derive an expression for \(\optdelta(\{k_i\}_{i=1}^n, \{(\varepsilon_k,\delta_k)\}_{k=1}^{k_n}, \mathcal{E})\) by showing that it suffices to analyze a class of randomized response mechanisms and to then compute the optimal ODP bound for these randomized response mechansisms. Kairouz et al.\ show the following lemma:
\begin{lemma}[\cite{kairouz2017composition}]
\label{thm:rr_equivalence}
For \(\varepsilon,\delta\geq 0\), let the randomized response mechanism \(\tilde{M}_{(\varepsilon,\delta)}:\{0,1\}\rightarrow\{0,1,2,3\}\) be defined as (dropping the dependency on \((\varepsilon,\delta)\) for simplicity)
\begin{align*}
    &\Pr(\tilde{M}(0) = 0) = \delta &&\Pr(\tilde{M}(1) = 0) = 0\\
    &\Pr(\tilde{M}(0) = 1) = \frac{(1-\delta)e^{\varepsilon}}{1+e^{\varepsilon}} &&\Pr(\tilde{M}(1) = 1) = \frac{(1-\delta)}{1+e^{\varepsilon}}\\
    &\Pr(\tilde{M}(0) = 2) = \frac{(1-\delta)}{1+e^{\varepsilon}} &&\Pr(\tilde{M}(1) = 2) = \frac{(1-\delta)e^{\varepsilon}}{1+e^{\varepsilon}}\\
    &\Pr(\tilde{M}(0) = 3) = 0 &&\Pr(\tilde{M}(1) = 3) = \delta.
\end{align*}
Then for any mechanism \(M\) that is \((\varepsilon,\delta)\)\hyp DP and any pair of neighboring databases \(\x^0,\x^1\) there exists a function \(T\) such that \(T(\tilde{M}_{(\varepsilon,\delta)}(b))\) is identically distributed to \(M(\x^b)\) for \(b=0,1\).
\end{lemma}

Based on this result, we show that it suffices to compute the ODP guarantee of a mechanism whose iterations consist of invocations of randomized response mechanisms, in order to compute \(\optdelta(\{k_i\}_{i=1}^n, \{(\varepsilon_k,\delta_k)\}_{k=1}^{k_n}, \mathcal{E}))\):
\begin{lemma}
\label{thm:rr_comp_equivalence}
For any \(\{k_i\}_{i=1}^n\), \(\{(\varepsilon_k,\delta_k)\}_{k=1}^{k_n}\) and any \(\mathcal{E}\) we have that
\begin{align*}
    &\optdelta(\{k_i\}_{i=1}^n, \{(\varepsilon_k,\delta_k)\}_{k=1}^{k_n}, \mathcal{E})\\
    &= \sup_{\{h^{k_i}\}_{i=1}^{n-1}} \optdelta(\{k_i\}_{i=1}^n, \{\tilde{M}_{(\varepsilon_k,\delta_k)}\}_{k=1}^{k_n}, \{h^{k_i}\}_{i=1}^{n-1}, \mathcal{E}).
\end{align*}
\end{lemma}

For the proof we need the following post\hyp processing lemma, whose proof follows that of the post\hyp processing lemma of DP \cite[Prop.~2.1]{dwork2014differential}:
\begin{lemma}
\label{thm:post-processing}
Let \(f=(f_1,\dots,f_{k_n})\) be a randomized function, where \(f_k\) may depend on the output of \(f_1,\dots,f_{k-1}\), and let
\begin{equation*}
    M=\im\big[\{k_i\}_{i=1}^n, \{M_k\}_{k=1}^{k_n}, \{h^{k_i}\circ (f_1,\dots,f_{k_i})\}_{i=1}^{n-1}\big]
\end{equation*}
be an iterative mechanism that fulfills \((\mathcal{P}^M=\{P_i^M\}_{i=1}^n,\mathcal{E},\delta)\)\hyp ODP. Assume that the output of each mechanism \(M_k\), \(k=2,\dots,k_n\), only depends on the database but not on the outputs of \(M_1,\dots,M_{k-1}\). Then the iterative mechanism
\begin{equation*}
    M'=\im\big[\{k_i\}_{i=1}^n, \{f_k\circ M_k\}_{k=1}^{k_n}, \{h^{k_i}\}_{i=1}^{n-1}\big]
\end{equation*}
fulfills \((\mathcal{P}^{M'}=\{P_i^{M'}\}_{i=1}^n,\mathcal{E}',\delta)\)\hyp ODP, where
\begin{equation*}
    \mathcal{E}'(P_i^{M'})=\mathcal{E}(P_i^M),
\end{equation*}
for \(i=1,\dots,n\).
\end{lemma}

\Lemmaref{thm:post-processing} allows us to prove \Lemmaref{thm:rr_comp_equivalence}, and with \Lemmaref{thm:rr_comp_equivalence} we can prove the main result of this subsection.
In the theorem and its proof we write \(q_{1,\dots,k}\) for the first \(k\) elements of a vector \(q\).

\begin{theorem}
\label{thm:iterative_opt_delta}
Let \(\{k_i\}_{i=1}^n\) be numbers of iterations, let \(\{(\varepsilon_k,\delta_k)\}_{k=1}^{k_n}\) be DP parameters and let \(\mathcal{E}\) be a function that assigns \(\varepsilon\)\hyp values to outputs of different lengths. For \(i=1,\dots,n\), define, for every \(Q\in\{0,1,2,3\}^{k_i}\) and for \(b=0,1\),
\begin{equation*}
    \Pr_b^{k_i}(Q) = \Pr\big((\tilde{M}_{(\varepsilon_1,\delta_1)}(b),\dots, \tilde{M}_{(\varepsilon_{k_i},\delta_{k_i})}(b)) \in Q\big).
\end{equation*}
For sets \(Q_{k_i}\in\{0,1,2,3\}^{k_i}\) and \(Q_{k_j}\in\{0,1,2,3\}^{k_j}\) with \(i<j\), write, with a slight abuse of notation, \(Q_{k_i}\cap Q_{k_j}=\emptyset\) if \(q_{1,\dots,k_i}\neq q'_{1,\dots,k_i}\) for all \(q\in Q_{k_i},\ q'\in Q_{k_j}\).
Then
\begin{align}
    &\optdelta(\{k_i\}_{i=1}^n, \{(\varepsilon_k,\delta_k)\}_{k=1}^{k_n}, \mathcal{E})\nonumber\\
    &= \max_{\substack{Q_{k_i}\in\{0,1,2,3\}^{k_i}\\
    i=1,\dots,n\\
    Q_{k_j}\cap Q_{k_l}=\emptyset\text{ for all } j<l}}
    \sum_{i=1}^n \big[\Pr_0^{k_i}(Q_{k_i}) - e^{\mathcal{E}(P_i)} \Pr_1^{k_i}(Q_{k_i})\big].\label{eq:iterative_opt_delta}
\end{align}
\end{theorem}
This quantity can be computed by iterating over the exponentially (in \(k_n\)) many possible sets \(Q_{k_i}\). Murtagh and Vadhan \cite{murtagh2016complexity} analyze the special case \(n=1\) and show that already that case is \(\#P\)\hyp complete. Thus, there is no hope for an efficient exact algorithm, but there might exist efficient approximation algorithms.

The main goal of this subsection is to show that one can get better ODP guarantees for iterative mechanisms than the ones from applying \Lemmaref{thm:odp_partition_iterative}. With \Lemmaref{thm:odp_partition_iterative}, we would get the following ODP\hyp\(\delta\):
\begin{align}
    \delta &= \sum_{i=1}^n \max_{Q_{k_i}\in\{0,1,2,3\}^{k_i}} \big[\Pr_0^{k_i}(Q_{k_i}) - e^{\mathcal{E}(P_i)} \Pr_1^{k_i}(Q_{k_i})\big]\nonumber\\
    &= \max_{\substack{Q_{k_i}\in\{0,1,2,3\}^{k_i}\\
    i=1,\dots,n}}
    \sum_{i=1}^n \big[\Pr_0^{k_i}(Q_{k_i}) - e^{\mathcal{E}(P_i)} \Pr_1^{k_i}(Q_{k_i})\big].\label{eq:iterative_nonopt_delta}
\end{align}
This is a direct application of \Lemmaref{thm:odp_partition_iterative} and the optimal composition theorem for adapative, heterogeneous mechanisms by Murtagh and Vadhan \cite{murtagh2016complexity} (without the simplifications of the expression performed by the authors).
Intuitively, \Thmref{thm:iterative_opt_delta} gives us an exact characterization of the sets of outputs that are possible, whereas in \eqnref{eq:iterative_nonopt_delta} we also need to take the maximum over impossible sets of outputs: let \(q_1,\dots,q_{k_i}\) be an output sequence such that the iterative mechanism consisting of a sequence of randomized response mechanisms terminates in iteration \(k_i\). Then it is impossible for the mechanism to output a sequence of length \(>k_i\) with the prefix \(q_1,\dots,q_{k_i}\).

This can lead to larger values for \(\delta\) in \eqnref{eq:iterative_nonopt_delta}.
For example, assume that \(\delta_k>0\) for all \(k=1,\dots,k_n\). We then have \(\Pr(\tilde{M}_{(\varepsilon_k,\delta_k)}(0)=0)=\delta_k > 0\) and \(\Pr(\tilde{M}_{(\varepsilon_k,\delta_k)}(1)=0)=0\) for all \(k=1,\dots,k_n\). Thus, each \(Q_{k_i}\) of the maximizer in \eqnref{eq:iterative_nonopt_delta} contains an output vector consisting of \(k_i\) \(\delta\)'s, whereas for the sets \(Q_{k_i}\) of the maximizer in \eqnref{eq:iterative_opt_delta} it must hold that at most one of them contains an output vector that only consists of \(\delta\)'s, leading to a strictly smaller maximum.

We leave a quantification of the size of the gap between \Lemmaref{thm:odp_partition_iterative} and \Thmref{thm:iterative_opt_delta} for future work, since this requires an efficient approximation algorithm for the maximization problem in \Thmref{thm:iterative_opt_delta}.

\subsection{Comparison with privacy filters}
What we propose --- stopping iterative mechanisms if they do not need more iterations, and thereby saving privacy budget that can be used on other queries --- can also be done with the privacy filters introduced by Rogers et al.\ \cite{rogers2016privacy}. In fact, a similar method for saving privacy budget when stopping early has recently been proposed in this context, though in combination with the related privacy odometers, which track privacy spending \cite{lecuyer2021practical}. The disadvantage of privacy filters is that they cannot simply use any existing DP composition result, but require their own composition theorems. It has been shown that Rényi DP (RDP) composition can be used for privacy filters \cite{feldman2021individual,lecuyer2021practical}, which is the currently best composition result for privacy filters. While RDP composition is a powerful tool, it does not always yield the best composition bounds \cite{bu2020deep}. With ODP, on the other hand, one can make use of any DP composition result via \Lemmaref{thm:odp_partition_iterative}, though at an additional cost in \(\delta\) that depends on the number \(n\) of potential stopping points and their positions.

\section{Recovering already spent privacy budget}
\label{sec:recovering_budget}
Assume that a data analyst invokes a DP mechanism \(M\) on a database \(\x\), which produces an output \(A\sim M(\x)\). For some reason the data analyst is not satisfied with the output. E.g., \(A\) could be a neural network that does not perform much better than random guessing, it could be a logistic regression model whose coefficients are not statistically significant, it could be some statistic with too large a confidence interval, etc. In all of these cases, \(A\) would be essentially useless and the data analyst would not publish it. Thus, it would be desirable if the analyst could get back the privacy budget that they spent on computing \(A\). For this, we could define a new mechanism \(\tilde{M}\) that first computes \(A\sim M(\x)\) and then invokes a test \(M_{\text{test}}\) on \(A\) that checks whether \(A\) should be released or not. If \(A\) should be released, \(\tilde{M}\) returns \(A\), otherwise \(\tilde{M}\) returns the symbol \(\bot\). In the first case, the data analyst would have to pay for the privacy cost of \(A\) and (potentially) the privacy cost of the test \(M_{\text{test}}\), in the second case only for the privacy cost of \(M_{\text{test}}\). \(M_{\text{test}}\) should thus be designed such that it uses up much less privacy budget than \(M\).

Formally, we want \(\tilde{M}\) to be \((\{\Range(M),\{\bot\}\},\mathcal{E},\delta)\)\hyp ODP with \(\mathcal{E}(\{\bot\})<\mathcal{E}(\Range(M))\). The test \(M_{\text{test}}\) is a function of the differentially private output \(A\) of \(M\), and in some cases also of a private database on which to evaluate this output. This second input is not necessary if the test can be performed on \(A\) directly. In the ML setting, there might be a private database \(\x\) that is split up into a training set \(\x^{\text{train}}\) that is used to train a differentially private model, and a test set \(\x^{\text{test}}\) on which \(M_{\text{test}}\) evaluates the model. In the case of a non\hyp\(\bot\)\hyp output, we only have to account for the privacy of \(M_{\text{test}}\) w.r.t.\ its database input \(\x^{\text{test}}\) but not w.r.t.\ \(\x^{\text{train}}\), since its first input, which is based on \(\x^{\text{train}}\), is the already differentially private output of \(M\). Since \(\x^{\text{train}}\) and \(\x^{\text{test}}\) are disjoint, we can compose the guarantees of \(M\) (w.r.t.\ \(\x^{\text{train}}\)) and of \(M_{\text{test}}\) (w.r.t.\ \(\x^{\text{test}}\)) via parallel composition. For the case of a \(\bot\)\hyp output, however, we want to exploit the fact that the output of \(M\) does not get revealed to obtain a better privacy guarantee. We thus need to compute a privacy guarantee of \(M_{\text{test}}(M(\x^{\text{train}}),\x^{\text{test}})\) w.r.t.\ both \(\x^{\text{train}}\) and \(\x^{\text{test}}\).

\subsection{Example: ERM}

\begin{algorithm}[tb]
\caption{\(\ERMOutput(\x^{\text{train}},\varepsilon,l,N,\Lambda)\)}
\label{alg:erm_output_perturbation}
\begin{algorithmic}[1]
\State{Sample random vector \(q\) with element\hyp wise density proportional to \(\exp\big(-\frac{n^{\text{train}}\Lambda\varepsilon}{2}\norm{q}_2\big)\)}
\State{Compute \(p_{\text{min}} = p_{\text{min}}(\x^{\text{train}})\) according to \eqnref{eq:erm_minimization}}
\State{\Return{\(\tilde{p}_{\text{min}} = p_{\text{min}} + q\)}}
\end{algorithmic}
\end{algorithm}

\begin{algorithm}[tb]
\caption{\(\LogRegTest(\x^{\text{train}},\x^{\text{test}},\varepsilon_1,\varepsilon_2,\Lambda,t)\)}
\label{alg:log_reg_with_test}
\begin{algorithmic}[1]
\State{\(\tilde{p}_{\text{min}} = \ERMOutput(\x^{\text{train}},\varepsilon_1,l_{\text{LogReg}},N_{L_2},\Lambda)\)}
\State{\(a = \max\left(\frac{2}{n^{\text{test}}},2\left(\exp\left(\frac{2}{n^{\text{train}}\Lambda}\right) - 1\right)\right)\)}
\State{Sample \(r\sim \Lap(a/\varepsilon_2)\)}
\If{\(s(\tilde{p}_{\text{min}},\x^{\text{test}}) + r \leq t\)}\label{alg_line:comparison_error}
    \State{\Return{\(\tilde{p}_{\text{min}}\)}}
\EndIf
\State{\Return{\(\bot\)}}
\end{algorithmic}
\end{algorithm}

As an example we consider the output perturbation mechanism by Chaudhuri et al.\ \cite{chaudhuri2011differentially} for releasing linear ML models with differential privacy. The class of models that their mechanism applies to includes, among others, logistic regression and (an approximation to) support vector machines (SVMs). Chaudhuri et al.\ assume covariates \(x\in\R^d\), \(\norm{x}_2\leq 1\), and labels \(y\in[-1,1]\). For a set \(\x^{\text{train}}\) of \(n^{\text{train}}\) training records they consider the empirical risk minimization problem
\begin{equation}
\label{eq:erm_minimization}
    p_{\text{min}}(\x^{\text{train}}) = \argmin_p \frac{1}{n^{\text{train}}}\sum_{(x,y)\in \x^{\text{train}}} l(y p^{\text{T}} x) + \Lambda N(p),
\end{equation}
where the minimization is over all parameter vectors \(p\); \(l\) is a loss function; and \(N\) a regularizer with regularization strength \(\Lambda\).

Chaudhuri et al.\ design \Algref{alg:erm_output_perturbation} for privately releasing the optimal parameter vector \(p\) and show that the algorithm fulfills DP:
\begin{theorem}
\label{thm:erm_output_perturbation}
If \(N\) is differentiable and \(1\)\hyp strongly convex, and \(l\) is convex and differentiable with \(\abs{l'(z)}\leq 1\) for all \(z\), then the \(L_2\)\hyp sensitivity of \(p_{\text{min}}\) is at most \(\frac{2}{n^{\text{train}}\Lambda}\), and \Algref{alg:erm_output_perturbation} is \(\varepsilon\)\hyp DP.
\end{theorem}

We extend \Algref{alg:erm_output_perturbation} with a differentially private test that checks whether the differentially private model returned by \Algref{alg:erm_output_perturbation} performs well enough for the respective application; only in that case will we release the model. For this we examplarily look at the case of logistic regression. The loss function of logistic regression is defined as
\begin{equation*}
    l_{\text{LogReg}}(z) = \ln(1 + e^{-z}),
\end{equation*}
often regularized with \(L_2\)\hyp regularization \(N_{L_2}(p) = \frac{1}{2} \norm{p}_2^2\).
A logistic regression model \(h_p\) with parameter vector \(p\) predicts the probability that the label of a record \(x\) is \(1\) as
\begin{equation*}
    h'_p(x) = \frac{1}{1+\exp(-p^{\text{T}} x)},
\end{equation*}
and the probability that the label is \(-1\) as \(1-h'_p(x)\). To map this output to the interval \([-1,1]\) we define
\begin{equation*}
    h_p(x) = 2 h'_p(x) - 1.
\end{equation*}
The prediction for the label of a record \(x\) is then typically \(\sign(h_p(x))\). Based on these label predictions it would be natural to use the accuracy of the model, i.e., the fraction of correct predictions, as the measure of model performance. However, if \(h_p(x)\) is close to \(0\) on all test records, then a small change in \(p\) might flip all label predictions, which implies that the sensitivity of the accuracy function is large. Since we want to perform a differentially private test on the model performance, we need a measure with smaller sensitivity. The error function \(s\) as defined below has this property.

Let \(\x^{\text{test}}\) be a private test set of size \(n^{\text{test}}\) that is disjoint from the training set. For our test we use the error function
\begin{equation*}
    s(p,\x^{\text{test}}) = \frac{1}{n^{\text{test}}} \sum_{(x,y)\in \x^{\text{test}}} \abs{h_p(x) - y}
\end{equation*}
--- i.e., the mean absolute error made by the logistic regression model ---, which can take values in \([0,2]\).
With this error function we define \Algref{alg:log_reg_with_test}, which computes a differentially private logistic regression parameter vector, but only releases this vector if the model error on the test set is not larger than a threshold \(t\).
In \Appref{sec:proof_log_reg_with_test} we prove the following theorem about the privacy of \Algref{alg:log_reg_with_test}:

\begin{theorem}
\label{thm:log_reg_with_test}
Let \(\varepsilon_1,\varepsilon_2>0\) and let
\begin{equation*}
    a = \max\left(\frac{2}{n^{\text{test}}},2\left(\exp\left(\frac{2}{n^{\text{train}}\Lambda}\right) - 1\right)\right).
\end{equation*}
Then \Algref{alg:log_reg_with_test} is \((\{\R^d,\{\bot\}\},\mathcal{E},0)\)\hyp ODP with
\begin{equation*}
    \mathcal{E}(\R^d) = \max\left(\varepsilon_1, \frac{2/n^{\text{test}}}{a}\varepsilon_2\right)
\end{equation*}
and
\begin{equation*}
    \mathcal{E}(\{\bot\}) = \min(\mathcal{E}(\R^d), \varepsilon_2)
\end{equation*}
w.r.t.\ a change in either \(\x^{\text{train}}\) or \(\x^{\text{test}}\).
\end{theorem}

In \Algref{alg:log_reg_with_test} we have a budget \(\varepsilon_1\) for the computation of the noisy parameter vector and a budget \(\varepsilon_2\) for the test. Since we want to save budget in the case of a \(\bot\)\hyp output, we will choose \(\varepsilon_2<\varepsilon_1\). This results in the following corollary:
\begin{corollary}
Let \(\varepsilon_1\geq\varepsilon_2>0\) and let
\begin{equation*}
    a = \max\left(\frac{2}{n^{\text{test}}},2\left(\exp\left(\frac{2}{n^{\text{train}}\Lambda}\right) - 1\right)\right).
\end{equation*}
Then \Algref{alg:log_reg_with_test} is \((\{\R^d,\{\bot\}\},\mathcal{E},0)\)\hyp ODP with
\begin{equation*}
    \mathcal{E}(\R^d) = \varepsilon_1
\end{equation*}
and
\begin{equation*}
    \mathcal{E}(\{\bot\}) = \varepsilon_2
\end{equation*}
w.r.t.\ a change in either \(\x^{\text{train}}\) or \(\x^{\text{test}}\).
\end{corollary}
\begin{proof}
The corollary follows from \Thmref{thm:log_reg_with_test} and from the fact that
\begin{equation*}
    \frac{2/n^{\text{test}}}{a} \leq 1.
\end{equation*}
\end{proof}

Thus, in the case where \(\varepsilon_1\geq\varepsilon_2\), we have the entire DP budget available for computing the noisy parameter vector when the test is passed. However, we pay indirectly for the test since we do not use the entire dataset for training, but reserve part of the records for testing.

\xhdr{Influence of noise on the test}
Since noise gets added to the value of the error function, \Algref{alg:log_reg_with_test} will not always return \(\bot\) and save privacy budget when the model's error exceeds the threshold \(t\). In \Figref{fig:erm} we plot the influence of this noise for \(\Lambda=1\), different values of \(\varepsilon_2\) and different sizes of the database. We assume that \(70\%\) of the records are used for training and \(30\%\) are used for testing. The \(y\)\hyp axis shows the \(95\)th percentile of the distribution of the noise that is added to the error. This can be interpreted as follows: if the error exceeds \(t\) plus the value of the \(95\)th percentile, then with probability at least \(95\%\) \Algref{alg:log_reg_with_test} will return \(\bot\), and thus an amount of \(\varepsilon_1-\varepsilon_2\) privacy budget will be saved as compared to releasing the noisy model parameters.

\section{Conclusions}
\label{sec:conclusions}
In this paper, we showed how to exploit the fact that certain DP mechanisms can produce outputs that leak less private information than other outputs. Specifically, we introduced a new composition experiment and {\em a posteriori} analysis techniques that lead to privacy budget savings under composition. We demonstrated the power of our techniques with three examples: (1) an improved analysis of the Sparse Vector Technique and the Propose-Test-Release framework; (2) accounting for the actual number of iterations performed by an iterative mechanism; (3) recovery of already spent privacy budget when a mechanism's output is unsatisfactory.

A data analyst can profit from our privacy budget savings by reducing the amount of noise that they need to add within mechanisms to provide DP, or by invoking more mechanisms.

Future research can utilize our techniques to improve the analysis of existing DP mechanisms under composition or to guide the design of new mechanisms that provide improved privacy--utility tradeoffs. Another promising direction is the combination of our {\em a posteriori} analyses with advanced composition analyses.

\begin{figure}[tb]
    \centering
    \includegraphics[width=0.48\textwidth]{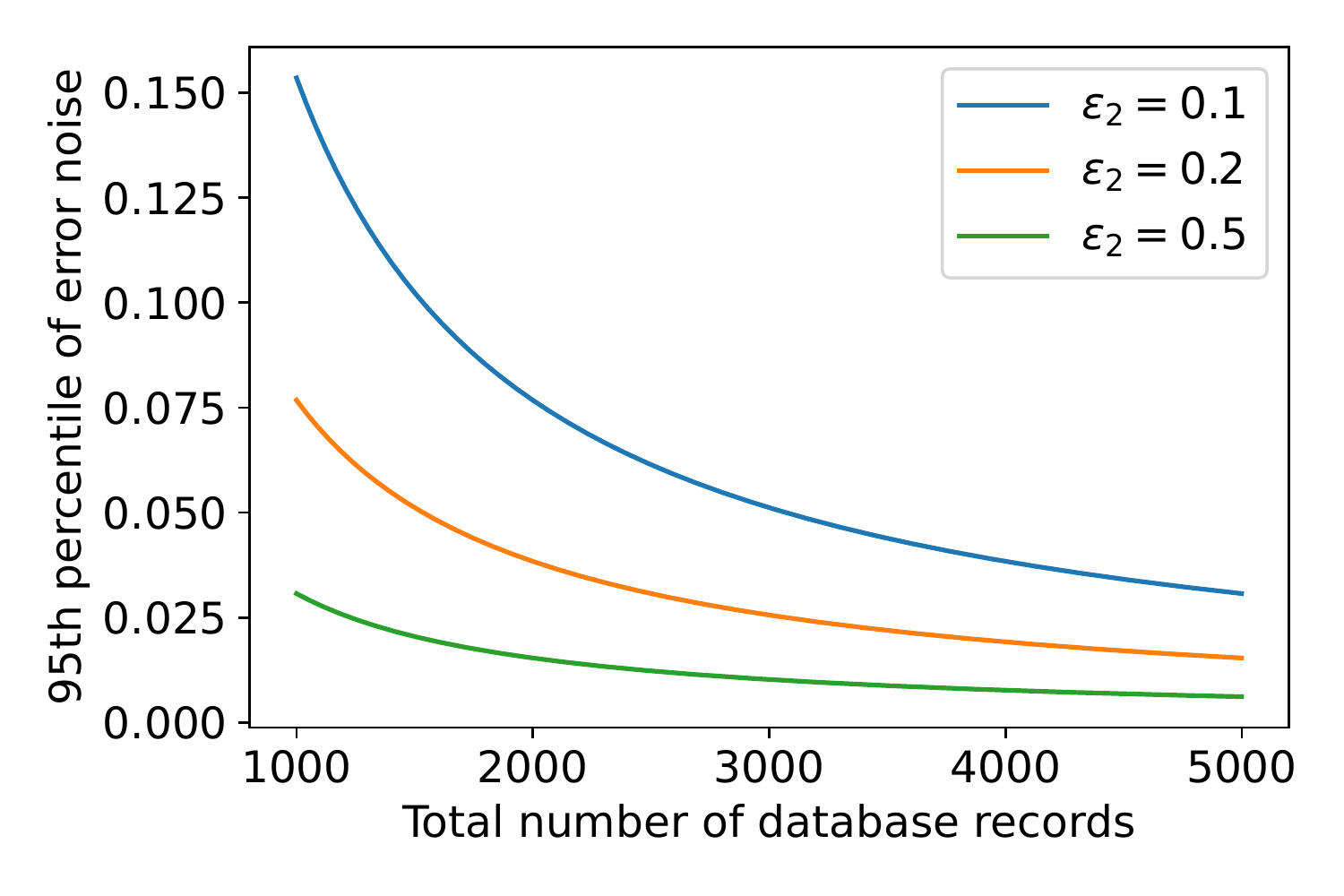}
    \caption{\(95\)th percentile of the noise added to the model's error for \(\Lambda=1\), different values of \(\varepsilon_2\) and different sizes of the database (of which \(70\%\) train, \(30\%\) test data). With probability at least \(95\%\) \Algref{alg:log_reg_with_test} will return \(\bot\) if the error exceeds \(t\) plus the \(95\)th percentile.}
    \label{fig:erm}
\end{figure}

\section*{Acknowledgments}
We would like to thank the anonymous reviewers and Borja Balle for their comments, which in particular helped improve the application examples for ODP.

Bentkamp's research has received funding from the European Research Council (ERC) under the European Union’s Horizon 2020 research and innovation program (grant agreement No. 713999, Matryoshka). It has also been funded by a Chinese Academy of Sciences President’s International Fellowship for Postdoctoral Researchers (grant No. 2021PT0015).
Bindschaedler's work is in part supported by the National Science Foundation under CNS-1933208. Any opinions, findings, and conclusions or recommendations expressed in this material are those of the authors and do not necessarily reflect the views of the National Science Foundation.


\bibliographystyle{abbrv}
\bibliography{references}

\appendix

\section{Towards advanced composition}
\label{sec:adv_comp}
\xhdr{Simple versus advanced composition}
The current ODP framework is suited for improving composition when a small to medium number of mechanisms is composed. When the number of mechanisms is large, though, it is more beneficial to use one of the advanced composition theorems (see \Secref{sec:related}).
Advanced composition decreases the \(\varepsilon\) in the DP guarantee when compared to simple composition, at the cost of increasing the \(\delta\)\hyp term. This increase in \(\delta\) is unacceptably high for short to medium\hyp length compositions, but asymptotically advanced composition is superior to simple composition. To provide sufficient privacy, a common recommendation is to choose \(\delta\ll 1/n\), where \(n\) is the size of the database.

We consider the setting where the privacy guarantees --- and thereby the utilities --- of the mechanisms that are composed are fixed and the data curator has a requirement on how large the value of \(\delta\) of their composition may at most be. In this setting the goal is to choose the composition theorem that yields the smallest \(\varepsilon\) while respecting the requirement on \(\delta\).
For advanced composition theorems such as the optimal advanced composition theorem for the setting of composing arbitrary mechanisms with the same DP guarantees \cite{kairouz2017composition} by Kairouz et al., \(\delta\) cannot be made arbitrarily small if a non\hyp zero improvement in \(\varepsilon\) over simple composition is required (\(i\) in their Thm.~3.3 needs to be at least \(1\) in that case). The shorter the length of the composition, the larger the minimal \(\delta\). Thus, the requirement on \(\delta\) can only be fulfilled if the composition is long enough; otherwise this composition theorem cannot be used.

\Tabref{table:composition_cutoff} shows the minimal length of the composition that is required to achieve a \(\delta\) not greater than a given value and an improvement in \(\varepsilon\) over simple composition when composing \(0.1\)\hyp DP mechanisms via the optimal advanced composition theorem by Kairouz et al.
If we require fewer compositions, then advanced composition does not at the same time yield an improvement in \(\varepsilon\) over simple composition --- or our composition protocol, which is based on simple composition--- and fulfill the requirement on \(\delta\). Since ODP composition improves upon simple composition, it is the best choice in terms of \(\varepsilon\) in this range of composition lengths whenever at least one of the composed mechanisms has a non\hyp trivial ODP guarantee. When the number of compositions is larger, then it depends on the concrete mechanisms whether ODP composition or advanced composition is better. Fixing the mechanisms and letting the number of compositions go to infinity, advanced composition is superior.

\begin{table*}[!tb]
\centering
\begin{tabular}{l|llllllll}
\(\delta\) & \(10^{-5}\) & \(10^{-6}\) & \(10^{-7}\) & \(10^{-8}\) & \(10^{-9}\) & \(10^{-10}\) & \(10^{-11}\) & \(10^{-12}\)\\\hline
\(I\) & \(17\) & \(20\) & \(24\) & \(27\) & \(31\) & \(35\) & \(38\) & \(42\)\\
\end{tabular}
\caption{The minimum number \(I\) of \(0.1\)\hyp DP mechanisms that need to be composed so that a \(\delta\)\hyp value of at most \(\delta\) and a smaller \(\varepsilon\)\hyp value than with simple composition is achieved when using the optimal composition theorem for homogeneous composition \cite{kairouz2017composition}. Below that number of composed mechanisms, advanced composition does not offer an advantage in terms of \(\varepsilon\) over simple composition, while ODP composition can offer an advantage over simple composition. A common recommendation is to choose \(\delta\ll 1/n\), where \(n\) is the size of the database,
or even cryptographically small (corresponding to roughly \(10^{-12}\)).}
\label{table:composition_cutoff}
\end{table*}

\xhdr{Advanced composition for ODP}
While in simple composition, on which our composition theorem is based, privacy in terms of \(\varepsilon\) degrades at a rate of \(\mathcal{O}(\varepsilon k)\), where \(k\) is the number of invoked mechanisms and \(\varepsilon\) the \(\varepsilon\)\hyp part of the DP guarantees of the composed mechanisms --- which we, for simplicity, assume is the same for all mechanisms ---, advanced composition theorems have a more gentle degradation rate of \(\mathcal{O}(\varepsilon^2 k + \varepsilon\sqrt{k})\) (but only yield better results for large enough \(k\)). Thus, it would be desirable to derive an advanced composition theorem in the framework of ODP. Privacy accounting in \Algref{alg:type1_comp} is done by subtracting after each mechanism invocation the amount of privacy budget that this mechanism has used up with its output. However, when we want privacy to degrade at a rate of \(\mathcal{O}(\varepsilon\sqrt{k})\) and want to keep the freedom for the adversary to freely choose mechanisms in each iteration whose privacy guarantees are not fixed beforehand, the amount of privacy degradation that was due to one particular mechanism is not known before \Algref{alg:type1_comp} has terminated.

For this reason, it makes sense to reformulate \Algref{alg:type1_comp} using a \emph{privacy filter}, as introduced by Rogers et al.\ \cite{rogers2016privacy}, which newly evaluates the entire sequence of mechanisms once a new mechanism gets invoked. In \Algref{alg:type1_comp} the adversary is only allowed to invoke mechanisms that fit within the remaining privacy budget. If the budget is used up, they may only invoke \((0,0)\)\hyp DP mechanisms. When using a privacy filter, we do not constrain the adversary in their choice of mechanisms, but instead stop the composition if invoking the mechanism chosen by the adversary would exceed the privacy budget. We define the new composition experiment in \Algref{alg:odp_filter}, and the privacy filter for ODP \(\mathcal{F}_{\varepsilon_t,\delta_t}\), which we term \emph{ODP filter}, as follows:
\begin{definition}[ODP Filter]
Fix \(\varepsilon_t,\delta_t\geq 0\). A function \(\mathcal{F}_{\varepsilon_t,\delta_t}:\R^{2I}_{\geq 0}\rightarrow \{\text{HALT},\text{CONT}\}\) is a \emph{valid ODP filter} if, for all adversaries \(\mathcal{A}\) and for all sets of views \(\mathcal{V}\) of \(\mathcal{A}\) returned by \Algref{alg:odp_filter}, we have that
\begin{equation*}
    \Pr(V^0\in\mathcal{V})\leq e^{\varepsilon_t}\Pr(V^1\in\mathcal{V}) + \delta_t.
\end{equation*}
\end{definition}

\Algref{alg:type1_comp} is equivalent to a \Algref{alg:odp_filter} with the ODP filter that receives as input a sequence \((\varepsilon^1,\delta_1,\dots,\varepsilon^I,\delta_I)\), and returns HALT if \(\sum_{i=1}^I \varepsilon^i > \varepsilon_t\) or \(\sum_{i=1}^I \delta^i > \delta_t\), and CONT otherwise. Due to \Thmref{thm:odpcomp}, this is a valid ODP filter.
In future work, it would be interesting to explore whether there exist valid ODP filters that have an asymptotically better than linear privacy degradation, as it is the case for privacy filters in the non\hyp ODP setting.
As an argument for why this might be the case, consider the mechanism \(M_{\text{toy}}\) from the introduction. Recall that \(M_{\text{toy}}\) flips a coin, and based on the result either invokes the Laplace mechanism on a function of the database, or returns the symbol \(\bot\), i.e., does not reveal any information about the database. If \(M_{\text{toy}}\) is invoked \(k\) times, and in \(k'\) of the cases the output is \(\bot\), then intuitively not more information about the database should have been revealed than when invoking the Laplace mechanism \(k-k'\) times without returning \(\bot\) between some of the invocations. An ODP filter that applies an advanced composition theorem to the composition of the \(k-k'\) Laplace mechanisms and only returns HALT if the resulting composition bound exceeds \((\varepsilon_t,\delta_t)\) should thus be valid.

\begin{algorithm}[!tbh]
\caption{\(\odpfc(\mathcal{A},\mathcal{F}_{\varepsilon_t,\delta_t},b)\)}
\label{alg:odp_filter}
\begin{algorithmic}[1]
\State {Select coin tosses \(R^b_\mathcal{A}\) for \(\mathcal{A}\) uniformly at random.}
\For{\(i=1,\dots,I\)}
    \State{\(\mathcal{A} = \mathcal{A}(R^b_\mathcal{A},\{A^b_j\}_{j=1}^i)\) chooses}
    \StateCont{~\(\bullet\)~ neighboring databases \(\x^{i,0},\x^{i,1}\),}
    \StateCont{~\(\bullet\)~ a triple \((\mathcal{P}_i=\{P_{i,k}\}_{k\in\mathcal{K}_i},\mathcal{E}_i,\delta_i)\), and}
    \StateCont{~\(\bullet\)~ a mechanism \(M_i\) that is \((\mathcal{P}_i,\mathcal{E}_i,\delta_i)\)\hyp ODP}
    \If{\(\mathcal{F}_{\varepsilon_t,\delta_t}(\varepsilon^1,\delta_1,\dots,\varepsilon^{i-1},\delta_{i-1},\sup_{P\in\mathcal{P}_i} \mathcal{E}_i(P),\delta_i,\)\par
        \hskip\algorithmicindent\(0,0,\dots,0,0)=\text{HALT}\)}
        \State{\(A^b_{i},\dots,A^b_I=\bot\)}
        \State{BREAK}
    \EndIf{}
    \State{Sample \(A^b_i = M_i(\x^{i,b})\)}
    \State{Let \(k\) such that \(A^b_i\in P_{i,k}\)}
    \State{Let \(\varepsilon^i=\mathcal{E}_i(P_{i,k})\)}
    \State{\(\mathcal{A}\) receives \(A^b_i\)}
\EndFor
\State{\Return{view \(V^b = (R^b_\mathcal{A}, A^b_1, \dots, A^b_I)\)}}
\end{algorithmic}
\end{algorithm}

\section{Proof of \Thmref{thm:odpcomp}}
\label{sec:proof_main_thm}
\begin{lemma} \label{lem:induction-step}
Let \(\x^0\) and \(\x^1\) be neighboring databases.
Let \(\varepsilon_t\geq 0\) and \(\delta_t \geq 0\).
Let \(M_1\) be a \((\{P_k\}_{k\in\mathcal{K}},\mathcal{E},\delta)\)-output differentially private mechanism with output set \(\mathcal{O}_1\)
such that \(\mathcal{E}(P_k)\leq\varepsilon_t\) for all \(k\)\ and \(\delta\leq\delta_t\).
For each \(v\in \mathcal{O}_1\), let \(M_2^0(v)\) and \(M_2^1(v)\) be random variables with values in 
a set \(\mathcal{O}_2\) such that
for all subsets \(\mathcal{V} \subseteq \mathcal{O}_2\)
and for \(v \in P_k\),
\begin{equation*}
    \Pr(M_2^0(v)\in \mathcal{V})
    \leq e^{\varepsilon_t - \mathcal{E}(P_k)} \Pr(M_2^1(v)\in \mathcal{V}) + \delta_t - \delta.
\end{equation*}
Then, for all \(\mathcal{V} \subseteq \mathcal{O}_1\times \mathcal{O}_2\), we have
\begin{align*}
    &\Pr((M_1(\x^0), M_2^0(M_1(\x^0)))\in \mathcal{V}) \\
    &\leq e^{\varepsilon_t} \Pr((M_1(\x^1), M_2^1(M_1(\x^1)))\in \mathcal{V}) + \delta_t.
\end{align*}
\end{lemma}

Note that the two occurrences of \(M_1(\x^0)\) in the equation refer to the same random variable, and likewise for the two occurrences of \(M_1(\x^1)\).

\begin{proof}
Let \(\mathcal{V} \subseteq \mathcal{O}_1\times \mathcal{O}_2\).
We partition \(\mathcal{V}\) into slices
\(\mathcal{V}_k = \mathcal{V} \cap (P_k \times \mathcal{O}_2)\) for~\(k\in\mathcal{K}\).
In the following, we roughly follow a proof of the simple composition theorem for \((\varepsilon,\delta)\)\hyp mechanisms by Dwork and Lei \cite[Lemma~28]{dwork2009differential}.

We define the signed measure
\begin{equation*}
    \mu(S) = \Pr(M_1(\x^0)\in S) - \sum_{k\in\mathcal{K}} e^{\mathcal{E}(P_k)} \Pr(M_1(\x^1)\in S \cap P_k).
\end{equation*}
We denote the positive measure resulting from its Hahn decomposition by \(\mu_+\).
This implies that for all \(S \subseteq \mathcal{O}_1\),
\begin{align*}
    &\Pr(M_1(\x^0)\in S)\\
    &\leq \sum_{k\in\mathcal{K}} e^{\mathcal{E}(P_k)} \Pr(M_1(\x^1)\in S \cap P_k) + \mu_+(S).
\end{align*}
For \(k\in\mathcal{K}\) and \(S \subseteq P_k\), it follows that
\begin{equation} \label{eq:pos_hahn_slice_prop}
    \Pr(M_1(\x^0)\in S) \leq e^{\mathcal{E}(P_k)} \Pr(M_1(\x^1)\in S ) + \mu_+(S).
\end{equation}\

By our assumption on \(M^0_2\) and \(M^1_2\), for every \(v\in P_k\) and
every \(S \subseteq O_1 \times O_2\),
\begin{align}\label{eq:diff_private_composition_min}
    &\Pr((v,M^0_2(v))\in S) = \Pr((v,M^0_2(v))\in S) \land 1\nonumber\\
    &\leq [e^{\varepsilon_t-\mathcal{E}(P_k)} \Pr((v,M^1_2(v))\in S) + \delta_t -\delta] \land 1\nonumber\\
    &\leq [e^{\varepsilon_t-\mathcal{E}(P_k)} \Pr((v,M^1_2(v))\in S)] \land 1 + \delta_t - \delta,
\end{align}
where \(\land\) denotes the minimum operator.

First, we consider a single partition \(P_k\) and its corresponding slice $\mathcal{V}_k$.
We write $V^0$ for $(M_1(\x^0),M^0_2(M_1(\x^0)))$ and 
$V^1$ for $(M_1(\x^1),M^1_2(M_1(\x^1)))$. We have
\begingroup
\allowdisplaybreaks 
\begin{align}
\nonumber&\Pr(V^0\in \mathcal{V}_k)\\
\nonumber&=\Pr((M_1(\x^0),M^0_2(M_1(\x^0)))\in \mathcal{V}_k)\\
\nonumber&= \int_{P_k} \Pr((v,M^0_2(v))\in \mathcal{V}_k) \Pr(M_1(\x^0)\in dv)\\
\nonumber&\overset{\text{(\ref{eq:diff_private_composition_min})}}{\leq}
    \int_{P_k} ([e^{\varepsilon_t - \mathcal{E}(P_k)} \Pr((v,M_2^1(v))\in \mathcal{V}_k)] \land 1 + \delta_t -\delta) \\[-\jot]
\nonumber&\hphantom{\leq}\times \Pr(M_1(\x^0)\in dv)\\
\nonumber&= \int_{P_k} ([e^{\varepsilon_t - \mathcal{E}(P_k)} \Pr((v,M^1_2(v))\in \mathcal{V}_k)] \land 1)\\[-\jot]
\nonumber&\hphantom{\leq}\times\Pr(M_1(\x^0)\in dv) + (\delta_t -\delta) \Pr(M_1(\x^0)\in P_k)\\
\nonumber&\overset{\text{(\ref{eq:pos_hahn_slice_prop})}}{\leq}
    \int_{P_k} ([e^{\varepsilon_t - \mathcal{E}(P_k)} \Pr((v,M^1_2(v))\in \mathcal{V}_k)] \land 1)\\[-\jot]
\nonumber&\hphantom{\leq}\times
    (e^{\mathcal{E}(P_k)} \Pr(M_1(\x^1)\in dv) + \mu_+(dv))\\[-\jot]
\nonumber&\hphantom{\leq}
    + (\delta_t -\delta) \Pr(M_1(\x^0)\in P_k)\\
\nonumber&\leq e^{\varepsilon_t} \int_{P_k} \Pr((v,M^1_2(v))\in \mathcal{V}_k) \Pr(M_1(\x^1)\in dv)\\[-\jot]
\nonumber&\hphantom{\leq}+ \mu_+(P_k) + (\delta_t -\delta) \Pr(M_1(\x^0)\in P_k)\\
\nonumber&\leq e^{\varepsilon_t} \Pr((M_1(\x^1),M^1_2(M_1(\x^1))\in \mathcal{V}_k)\\[-\jot]
\nonumber&\hphantom{\leq}+ \mu_+(P_k) + (\delta_t -\delta) \Pr(M_1(\x^0)\in P_k)\\
\nonumber&= e^{\varepsilon_t} \Pr(V^1\in \mathcal{V}_k)
    + \mu_+(P_k)\\[-\jot]
&\hphantom{=}+ (\delta_t -\delta) \Pr(M_1(\x^0)\in P_k).
    \label{eq:slice}
\end{align}
\endgroup

Next, we consider all partitions \(P_k\) and slices \(\mathcal{V}_k\) together.
Since \(M_1\) is \((\{P_k\}_{k\in\mathcal{K}},\mathcal{E},\delta)\)-ODP,
for all $S\subseteq\mathcal{O}_1$, we have
\begin{equation*}
\Pr(M_1(\x^0) \in S) \leq \delta + \sum_{k\in\mathcal{K}} e^{\mathcal{E}(P_k)} \Pr(M_1(\x^1) \in S\cap P_k),
\end{equation*}
and thus $\mu(S) \leq \delta$. It follows that for all  $S\subseteq\mathcal{O}_1$,
\begin{equation}\label{eq:pos_hahn_prop}
\mu_+(S) \leq  \delta.
\end{equation}
Thus,
\begingroup
\allowdisplaybreaks 
\begin{align*} 
    &\Pr(V^0\in\mathcal{V})
    = \Pr(V^0\in\bigcup_k \mathcal{V}_k)
    = \sum_k \Pr(V^0\in\mathcal{V}_k)\\
    &\overset{\text{(\ref{eq:slice})}}{\leq} 
    \sum_k [e^{\varepsilon_t} \Pr(V^1\in \mathcal{V}_k) + \mu_+(P_k)\\
    &\hphantom{\overset{\text{(\ref{eq:slice})}}{\leq} 
    \sum_k [}+ (\delta_t - \delta) \Pr(M_1(\x^0)\in P_k)]\\
    &= e^{\varepsilon_t} \sum_k \Pr(V^1\in \mathcal{V}_k) + \sum_k \mu_+(P_k) 
    \\[-\jot]
    &\hphantom{\leq}+ \sum_k (\delta_t - \delta) \Pr(M_1(\x^0)\in P_k)\\
    &= e^{\varepsilon_t} \Pr(V^1\in\mathcal{V}) + \mu_+(\mathcal{O}_1)
    + (\delta_t - \delta) \Pr(M_1(\x^0)\in \mathcal{O}_1)\\
    &\overset{\text{(\ref{eq:pos_hahn_prop})}}{\leq} 
    e^{\varepsilon_t} \Pr(V^1\in\mathcal{V}) + \delta + (\delta_t - \delta) \times 1\\
    &= e^{\varepsilon_t} \Pr(V^1\in\mathcal{V}) + \delta_t.
    \tag*{\qedhere}
\end{align*}
\endgroup

\end{proof}

\newcounter{savethmcounter}
\setcounter{savethmcounter}{\value{thmcounter}}
\setcounterref{thmcounter}{thm:odpcomp}
\addtocounter{thmcounter}{-1} 
\begin{theorem}
For every adversary \(\mathcal{A}\) and for every set of views \(\mathcal{V}\) of \(\mathcal{A}\) returned by \Algref{alg:type1_comp} we have that
\begin{equation*}
    \Pr(V^0\in\mathcal{V}) \leq e^{\varepsilon_t} \Pr(V^1\in\mathcal{V}) + \delta_t.
\end{equation*}
\end{theorem}
\setcounter{thmcounter}{\value{savethmcounter}}

\begin{proof}

Since the adversary's coins are tossed before any computation on the data is done, we can fix the randomness \(R^b_\mathcal{A}\) of the adversary and thus assume a deterministic adversary.

We proceed by induction on the number of iterations \(I\) of \Algref{alg:type1_comp}.
The theorem holds for \(I=1\) by Lemma~\ref{thm:odp_to_dp}.
By the induction hypothesis,
we can assume that the theorem holds for \(I\) iterations and
we must show that it holds for \(I+1\) iterations.

Since the adversary is deterministic, 
the databases \(\x^{1,0}\) and \(\x^{1,1}\), the triple \((\mathcal{P}_1 = \{P_{1,k}\}_{k\in\mathcal{K}_1},\mathcal{E}_1,\delta_1)\), and the ODP mechanism \(M_1\) are fixed.

Let $v$ be a possible output of $M_1$.
Let $\mathcal{A}_r(v)$ be the adversary 
that behaves like adversary $\mathcal{A}$ after adversary $\mathcal{A}$ has seen $M_1$
produce output $v$.
Let $k \in \mathcal{K}_1$ such that $v \in P_{1,k}$.
We invoke the induction hypothesis using 
$\mathcal{A}_r(v)$ for $\mathcal{A}$,
$\varepsilon_t - \mathcal{E}_1(P_{1,k})$ for $\varepsilon_t$,
and
$\delta_t - \delta_1$ for $\delta_t$
to obtain
\begin{equation*}
  \Pr(V^0_r(v)\in\mathcal{V}_r) 
  \leq e^{\varepsilon_t - \mathcal{E}_1(P_{1,k})} \Pr(V^1_r(v)\in\mathcal{V}_r) + \delta_t - \delta_1,
\end{equation*}
for every set of views $\mathcal{V}_r$ of $\mathcal{A}_r(v)$ returned by \Algref{alg:type1_comp},
where $V^b_r(v)$ is the random variable describing the view of $\mathcal{A}_r(v)$ returned by \Algref{alg:type1_comp} for bit $b$ with 
total budgets $\varepsilon_t - \mathcal{E}_1(P_{1,k})$ and $\delta_t - \delta_1$.

We can thus apply Lemma \ref{lem:induction-step}
using $M_1$ for $M_1$ and $V^b_r(v)$ for $M_2^b(v)$, yielding
\begin{align*}
    &\Pr((M_1(\x^{1,0}), V^0_r(M_1(\x^{1,0})))\in \mathcal{V}) \\
    &\leq e^{\varepsilon_t} \Pr((M_1(\x^{1,1}), V^1_r(M_1(\x^{1,1})))\in \mathcal{V}) + \delta_t.
\end{align*}
Since $V^b = (M_1(\x^{1,b}), V^b_r(M_1(\x^{1,b})))$, this completes the induction step and the proof.
\end{proof}

\section{Proofs from \Secref{sec:comp_direct}}
\newcounter{savethmcounter_app_iteration_length}
\setcounter{savethmcounter_app_iteration_length}{\value{thmcounter}}

We first prove \Lemmaref{thm:post-processing}, which we use in the proof of \Lemmaref{thm:rr_comp_equivalence}.

\setcounterref{thmcounter}{thm:post-processing}
\addtocounter{thmcounter}{-1}
\begin{lemma}
Let \(f=(f_1,\dots,f_{k_n})\) be a randomized function, where \(f_k\) may depend on the output of \(f_1,\dots,f_{k-1}\), and let \(M=\im\big[\{k_i\}_{i=1}^n, \{M_k\}_{k=1}^{k_n}, \{h^{k_i}\circ (f_1,\dots,f_{k_i})\}_{i=1}^{n-1}\big]\) be an iterative mechanism that fulfills \((\mathcal{P}^M=\{P_i^M\}_{i=1}^n,\mathcal{E},\delta)\)\hyp ODP. Assume that the outputs of each mechanism \(M_k\), \(k=2,\dots,k_n\), only depends on the database but not on the outputs of \(M_1,\dots,M_{k-1}\). Then the iterative mechanism \(M'=\im\big[\{k_i\}_{i=1}^n, \{f_k\circ M_k\}_{k=1}^{k_n}, \{h^{k_i}\}_{i=1}^{n-1}\big]\) fulfills \((\mathcal{P}^{M'}=\{P_i^{M'}\}_{i=1}^n,\mathcal{E}',\delta)\)\hyp ODP, where, for \(i=1,\dots,n\), \(\mathcal{E}'(P_i^{M'})=\mathcal{E}(P_i^M)\).
\end{lemma}
\begin{proof}
First assume that \(f\) is deterministic. Let \(\x,\xp\) be neighboring databases and let \(S'\subseteq\Range(M')\). Let
\begin{equation*}
    S=\{s\in\Range(M)\mid (f_1,\dots,f_{\dim(s)})(s) \in S'\},
\end{equation*}
where \(\dim(s)\) denotes the dimensionality of the vector \(s\).
Then
\begin{align*}
    &\Pr(M'(\x)\in S') = \Pr(M(\x)\in S)\\
    &\leq \delta + \sum_{i=1}^n e^{\mathcal{E}(P_i^M)} \Pr(M(\x)\in S\cap P_i^M)\\
    &= \delta + \sum_{i=1}^n e^{\mathcal{E}(P_i^M)} \Pr(M'(\x)\in S'\cap P_i^{M'}).
\end{align*}
This proves the statement for deterministic functions \(f\). Every randomized function \(f\) can be written as a random convex combination of deterministic functions. Since every convex combination of \((\mathcal{P},\mathcal{E},\delta)\)\hyp ODP mechanisms fulfills \((\mathcal{P},\mathcal{E},\delta)\)\hyp ODP, the statement follows for randomized functions \(f\).
\end{proof}

\setcounterref{thmcounter}{thm:rr_comp_equivalence}
\addtocounter{thmcounter}{-1}
\begin{lemma}
For any \(\{k_i\}_{i=1}^n\), \(\{(\varepsilon_k,\delta_k)\}_{k=1}^{k_n}\) and any \(\mathcal{E}\) we have that
\begin{align*}
    &\optdelta(\{k_i\}_{i=1}^n, \{(\varepsilon_k,\delta_k)\}_{k=1}^{k_n}, \mathcal{E})\\
    &= \sup_{\{h^{k_i}\}_{i=1}^{n-1}} \optdelta(\{k_i\}_{i=1}^n, \{\tilde{M}_{(\varepsilon_k,\delta_k)}\}_{k=1}^{k_n}, \{h^{k_i}\}_{i=1}^{n-1}, \mathcal{E}).
\end{align*}
\end{lemma}
\begin{proof}[Proof of \Lemmaref{thm:rr_comp_equivalence}]
Since, for \(k=1,\dots,k_n\), \(\tilde{M}_{(\varepsilon_k,\delta_k)}\) fulfills \((\varepsilon_k,\delta_k)\)\hyp DP, we have that
\begin{align*}
    &\optdelta(\{k_i\}_{i=1}^n, \{(\varepsilon_k,\delta_k)\}_{k=1}^{k_n}, \mathcal{E})\\
    &= \sup_{\substack{\{M_k\}_{k=1}^{k_n}\\ \{h^{k_i}\}_{i=1}^{n-1}}} \{\optdelta(\{k_i\}_{i=1}^n, \{M_k\}_{k=1}^{k_n}, \{h^{k_i}\}_{i=1}^{n-1}, \mathcal{E})\\
    &\hphantom{= \sup_{\substack{\{M_k\}_{k=1}^{k_n}\\ \{h^{k_i}\}_{i=1}^{n-1}}} \{}
    \mid M_k \text{ is } (\varepsilon_k,\delta_k)\text{-DP}, k=1,\dots,k_n\}\\
    &\geq \sup_{\{h^{k_i}\}_{i=1}^{n-1}} \optdelta(\{k_i\}_{i=1}^n, \{\tilde{M}_{(\varepsilon_k,\delta_k)}\}_{k=1}^{k_n}, \{h^{k_i}\}_{i=1}^{n-1}, \mathcal{E}).
\end{align*}
For the converse, define an iterative algorithm \(M=\im\big[\{k_i\}_{i=1}^n, \{M_k\}_{k=1}^{k_n}, \{h^{k_i}\}_{i=1}^{n-1}\big]\), where, for \(k=1,\dots,k_n\), \(M_k\) fulfills \((\varepsilon_k,\delta_k)\)\hyp DP. Since we allow adaptive composition, \(M_k\) may depend on the outputs of \(M_1,\dots,M_{k-1}\), which we denote by \(s_1,\dots,s_{k-1}\). We write \(M_k^{s_1,\dots,s_{k-1}}\) when we want to make this dependency explicit. Fix a pair of neighboring databases \(\x^0,\x^1\). Due to \Lemmaref{thm:rr_equivalence} there exists, for every \(k\) and for every sequence of previous outputs \(s_1,\dots,s_{k-1}\), a function \(T_k^{s_1,\dots,s_{k-1}}\) such that \(T_k^{s_1,\dots,s_{k-1}}(\tilde{M}_{(\varepsilon_k,\delta_k)}(b))\) follows the same distribution as \(M_k^{s_1,\dots,s_{k-1}}(\x^b)\) for \(b=0,1\). In following Murtagh and Vadhan, we define a function \(\hat{T}(z_1,\dots,z_{k_n})\), where \(z_1,\dots,z_{k_n}\in\{0,1,2,3\}\), as follows:
\begin{equation*}
    \hat{T}_k(z_k) = T_k^{T_1(z_1),\dots,T_{k-1}(z_{k-1})}(z_k).
\end{equation*}
Then, for \(b=0,1\),
\begin{equation*}
    M(\x^b)=\im\big[\{k_i\}_{i=1}^n, \{M_k\}_{k=1}^{k_n}, \{h^{k_i}\}_{i=1}^{n-1}\big](\x^b)
\end{equation*}
follows the same distribution as
\begin{align*}
    M'(\x^b)\defeq \im\big[&\{k_i\}_{i=1}^n, \{\hat{T}_k\circ \tilde{M}_{(\varepsilon_k,\delta_k)}\}_{k=1}^{k_n},\\
    &\{h^{k_i}\}_{i=1}^{n-1}\big](\x^b).
\end{align*}
Let
\begin{align*}
    M''=\im\big[&\{k_i\}_{i=1}^n, \{\tilde{M}_{(\varepsilon_k,\delta_k)}\}_{k=1}^{k_n},\\
    &\{h^{k_i}\circ (\hat{T}_1,\dots,\hat{T}_{k_i})\}_{i=1}^{n-1}\big].
\end{align*}
From \Lemmaref{thm:rr_comp_equivalence} it follows that any ODP guarantee that holds for \(M''\) also hold for \(M'\) and thus for \(M\). Hence,
\begin{align*}
    &\optdelta(\{k_i\}_{i=1}^n, \{M_k\}_{k=1}^{k_n}, \{h^{k_i}\}_{i=1}^{n-1}, \mathcal{E})\\
    &\leq \optdelta(\{k_i\}_{i=1}^n, \{\tilde{M}_{(\varepsilon_k,\delta_k)}\}_{k=1}^{k_n},\\
    &\hphantom{\leq \optdelta(}\{h^{k_i}\circ (\hat{T}_1,\dots,\hat{T}_{k_i})\}_{i=1}^{n-1}, \mathcal{E}).
\end{align*}
Taking the supremum over the mechanisms and the stopping criteria on both sides yields
\begin{align*}
    &\optdelta(\{k_i\}_{i=1}^n, \{(\varepsilon_k,\delta_k)\}_{k=1}^{k_n}, \mathcal{E})\\
    &= \sup_{\substack{\{M_k\}_{k=1}^{k_n}\\ \{h^{k_i}\}_{i=1}^{n-1}}} \{\optdelta(\{k_i\}_{i=1}^n, \{M_k\}_{k=1}^{k_n}, \{h^{k_i}\}_{i=1}^{n-1}, \mathcal{E})\\
    &\hphantom{= \sup_{\substack{\{M_k\}_{k=1}^{k_n}\\ \{h^{k_i}\}_{i=1}^{n-1}}} \{}
    \mid M_k \text{ is } (\varepsilon_k,\delta_k)\text{-DP}, k=1,\dots,k_n\}\\
    &\leq \sup_{\hat{T},\{h^{k_i}\}_{i=1}^{n-1}} \optdelta(\{k_i\}_{i=1}^n, \{\tilde{M}_{(\varepsilon_k,\delta_k)}\}_{k=1}^{k_n},\\
    &\hphantom{\leq \sup_{\hat{T},\{h^{k_i}\}_{i=1}^{n-1}} \optdelta(}
    \{h^{k_i}\circ (\hat{T}_1,\dots,\hat{T}_{k_i})\}_{i=1}^{n-1}, \mathcal{E})\\
    &= \sup_{\{h^{k_i}\}_{i=1}^{n-1}} \optdelta(\{k_i\}_{i=1}^n, \{\tilde{M}_{(\varepsilon_k,\delta_k)}\}_{k=1}^{k_n}, \{h^{k_i}\}_{i=1}^{n-1}, \mathcal{E}).
\end{align*}
\end{proof}

\setcounterref{thmcounter}{thm:iterative_opt_delta}
\addtocounter{thmcounter}{-1}
\begin{theorem}
Let \(\{k_i\}_{i=1}^n\) be numbers of iterations, let \(\{(\varepsilon_k,\delta_k)\}_{k=1}^{k_n}\) be DP parameters and let \(\mathcal{E}\) be a function that assigns \(\varepsilon\)\hyp values to outputs of different lengths. For \(i=1,\dots,n\), define, for every \(Q\in\{0,1,2,3\}^{k_i}\) and for \(b=0,1\),
\begin{equation*}
    \Pr_b^{k_i}(Q) = \Pr\big((\tilde{M}_{(\varepsilon_1,\delta_1)}(b),\dots, \tilde{M}_{(\varepsilon_{k_i},\delta_{k_i})}(b)) \in Q\big).
\end{equation*}
For sets \(Q_{k_i}\in\{0,1,2,3\}^{k_i}\) and \(Q_{k_j}\in\{0,1,2,3\}^{k_j}\) with \(i<j\), write, with a slight abuse of notation, \(Q_{k_i}\cap Q_{k_j}=\emptyset\) if \(q_{1,\dots,k_i}\neq q'_{1,\dots,k_i}\) for all \(q\in Q_{k_i},\ q'\in Q_{k_j}\).
Then
\begin{align*}
    &\optdelta(\{k_i\}_{i=1}^n, \{(\varepsilon_k,\delta_k)\}_{k=1}^{k_n}, \mathcal{E})\\
    &= \max_{\substack{Q_{k_i}\in\{0,1,2,3\}^{k_i}\\
    i=1,\dots,n\\
    Q_{k_j}\cap Q_{k_l}=\emptyset\text{ for all } j<l}}
    \sum_{i=1}^n \big[\Pr_0^{k_i}(Q_{k_i}) - e^{\mathcal{E}(P_i)} \Pr_1^{k_i}(Q_{k_i})\big].
\end{align*}
\end{theorem}
\begin{proof}
From \Lemmaref{thm:rr_comp_equivalence} we know that we only have to compute
\begin{equation}
    \label{eq:optdelta_rr}
    \sup_{\{h^{k_i}\}_{i=1}^{n-1}} \optdelta(\{k_i\}_{i=1}^n, \{\tilde{M}_{(\varepsilon_k,\delta_k)}\}_{k=1}^{k_n}, \{h^{k_i}\}_{i=1}^{n-1}, \mathcal{E}).
\end{equation}
Fix a sequence of (potentially randomized) stopping criteria \(\{\tilde{h}^{k_i}\}_{i=1}^{n-1}\). For \(q\in\{0,1,2,3\}^{k_i}\), write \(\Pr_b^{k_i}(q)=\Pr_b^{k_i}(\{q\})\) and
\begin{align*}
    &R(\{\tilde{h}^{k_j}\}_{j=1}^i,q)\\
    &= \Pr(\tilde{h}^{k_j}(q_{1,\dots,k_j})=0 \text{ for all } j<i,\ \tilde{h}^{k_i}(q)=1).
\end{align*}
Then \(\optdelta(\{k_i\}_{i=1}^n, \{\tilde{M}_{(\varepsilon_k,\delta_k)}\}_{k=1}^{k_n}, \{\tilde{h}^{k_i}\}_{i=1}^{n-1}, \mathcal{E})\) is the minimal value \(\tilde{\delta}\) such that, for all \(Q_{k_i}\in\{0,1,2,3\}^{k_i}\), \(i=1,\dots,n\),
\begin{align*}
    &\sum_{i=1}^n \sum_{q\in Q_{k_i}} R(\{\tilde{h}^{k_j}\}_{j=1}^i,q) \Pr_0^{k_i}(q)\\
    &\leq \tilde{\delta} + \sum_{i=1}^n \sum_{q\in Q_{k_i}} e^{\mathcal{E}(P_i)} R(\{\tilde{h}^{k_j}\}_{j=1}^i,q) \Pr_1^{k_i}(q),
\end{align*}
i.e.,
\begin{equation}
\label{eq:iterative_mechanism_delta_stopping}
    \sum_{i=1}^n \sum_{q\in Q_{k_i}} R(\{\tilde{h}^{k_j}\}_{j=1}^i,q) \big(\Pr_0^{k_i}(q) - e^{\mathcal{E}(P_i)} \Pr_1^{k_i}(q)\big)
    \leq \tilde{\delta}.
\end{equation}
Let \(H\) be the set of all sequences \(\{h^{k_i}\}_{i=1}^{n-1}\) of deterministic stopping criteria. Since the domain and the range of all stopping criteria are finite sets, \(H\) is also a finite set. We can thus write
\begin{align*}
    &\sum_{i=1}^n \sum_{q\in Q_{k_i}} R(\{\tilde{h}^{k_j}\}_{j=1}^i,q) \big(\Pr_0^{k_i}(q) - e^{\mathcal{E}(P_i)} \Pr_1^{k_i}(q)\big)\\
    &=\sum_{\{h^{k_i}\}_{i=1}^{n-1} \in H} \sum_{i=1}^n \sum_{q\in Q_{k_i}} \big[R(\{h^{k_j}\}_{j=1}^i,q)\\
    &\hphantom{=\big[} \Pr(\{\tilde{h}^{k_i}\}_{i=1}^{n-1} = \{h^{k_i}\}_{i=1}^{n-1}) \big(\Pr_0^{k_i}(q) - e^{\mathcal{E}(P_i)} \Pr_1^{k_i}(q)\big)\big]\\
    &\leq \sum_{\{h^{k_i}\}_{i=1}^{n-1} \in H} \sum_{i=1}^n \sum_{q\in Q_{k_i}} \big[\max_{\{\bar{h}^{k_i}\}_{i=1}^{n-1} \in H} R(\{\bar{h}^{k_j}\}_{j=1}^i,q)\\
    &\hphantom{=\big[} \Pr(\{\tilde{h}^{k_i}\}_{i=1}^{n-1} = \{h^{k_i}\}_{i=1}^{n-1}) \big(\Pr_0^{k_i}(q) - e^{\mathcal{E}(P_i)} \Pr_1^{k_i}(q)\big)\big]\\
    &= \sum_{i=1}^n \sum_{q\in Q_{k_i}} \big[\max_{\{\bar{h}^{k_i}\}_{i=1}^{n-1} \in H} R(\{\bar{h}^{k_j}\}_{j=1}^i,q)\\
    &\hphantom{= \sum_{i=1}^n \sum_{q\in Q_{k_i}} \big[} \big(\Pr_0^{k_i}(q) - e^{\mathcal{E}(P_i)} \Pr_1^{k_i}(q)\big)\\
    &\hphantom{= \sum_{i=1}^n \sum_{q\in Q_{k_i}} \big[} \sum_{\{h^{k_i}\}_{i=1}^{n-1} \in H} \Pr(\{\tilde{h}^{k_i}\}_{i=1}^{n-1} = \{h^{k_i}\}_{i=1}^{n-1})\big]\\
    &= \sum_{i=1}^n \sum_{q\in Q_{k_i}} \max_{\{\bar{h}^{k_i}\}_{i=1}^{n-1} \in H} R(\{\bar{h}^{k_j}\}_{j=1}^i,q)\\
    &\hphantom{= \sum_{i=1}^n \sum_{q\in Q_{k_i}}} \big(\Pr_0^{k_i}(q) - e^{\mathcal{E}(P_i)} \Pr_1^{k_i}(q)\big).
\end{align*}
Thus, the worst\hyp case in terms of \(\delta\) is achieved for a sequence of deterministic stopping criteria. In \eqnref{eq:optdelta_rr} we therefore only have to take the supremum over deterministic stopping criteria, and can replace it with a maximum, since the set \(H\) of these criteria is finite.
For deterministic stopping criteria \(\{h^{k_i}\}_{i=1}^{n-1}\), the values 
\begin{align*}
    &R(\{h^{k_j}\}_{j=1}^i,q)\\
    &= \Pr(h^{k_j}(q_{1,\dots,k_j})=0 \text{ for all } j<i,\ h^{k_i}(q)=1)
\end{align*}
are each either \(0\) or \(1\). We have, using \(\mathbf{1}\) as the indicator function:

\begingroup
\allowdisplaybreaks 
\begin{align}
    &\max_{\{h^{k_i}\}_{i=1}^{n-1}} \optdelta(\{k_i\}_{i=1}^n, \{\tilde{M}_{(\varepsilon_k,\delta_k)}\}_{k=1}^{k_n}, \{h^{k_i}\}_{i=1}^{n-1}, \mathcal{E})\nonumber\\
    &=\max_{\{h^{k_i}\}_{i=1}^{n-1}\in H}
    \max_{\substack{Q_{k_i}\in\{0,1,2,3\}^{k_i}\\
    i=1,\dots,n}}\nonumber\\
    &\hphantom{=} \sum_{i=1}^n \sum_{q\in Q_{k_i}} \big[\mathbf{1}(h^{k_j}(q_{1,\dots,k_j})=0\nonumber\\
    &\hphantom{= \sum_{i=1}^n \sum_{q\in Q_{k_i}} \big[} \text{for all } j<i,\ h^{k_i}(q)=1)\label{eq:optdelta_final_eq_1}\\
    &\hphantom{=\sum_{i=1}^n \sum_{q\in Q_{k_i}} \big[}\big(\Pr_0^{k_i}(q) - e^{\mathcal{E}(P_i)} \Pr_1^{k_i}(q)\big)\big]\nonumber\\
    &=\max_{\substack{Q_{k_i}\in\{0,1,2,3\}^{k_i}\\
    i=1,\dots,n\\
    Q_{k_j}\cap Q_{k_l}=\emptyset\text{ for all } j<l}}
    \sum_{i=1}^n \sum_{q\in Q_{k_i}}
    \big[\Pr_0^{k_i}(q) - e^{\mathcal{E}(P_i)} \Pr_1^{k_i}(q)\big]\label{eq:optdelta_final_eq_2}\\
    &=\max_{\substack{Q_{k_i}\in\{0,1,2,3\}^{k_i}\\
    i=1,\dots,n\\
    Q_{k_j}\cap Q_{k_l}=\emptyset\text{ for all } j<l}}
    \sum_{i=1}^n
    \big[\Pr_0^{k_i}(Q_{k_i}) - e^{\mathcal{E}(P_i)} \Pr_1^{k_i}(Q_{k_i})\big].\nonumber
\end{align}
\endgroup
The equality between \eqnref{eq:optdelta_final_eq_1} and \eqnref{eq:optdelta_final_eq_2} holds because of the following:

\(\text{(\ref{eq:optdelta_final_eq_1})}\leq\text{(\ref{eq:optdelta_final_eq_2})}\) since for every sequence of tests \(\{h^{k_i}\}_{i=1}^{n-1}\) and for every sequence of sets \(\{Q_{k_i}\}_{i=1}^n\) we get the same quantity in \eqnref{eq:optdelta_final_eq_2} as in \eqnref{eq:optdelta_final_eq_1} by choosing that sequence of subsets of the sets \(\{Q_{k_i}\}_{i=1}^n\) that excludes those elements \(q\) for which
\begin{equation*}
    \mathbf{1}(h^{k_j}(q_{1,\dots,k_j})=0 \text{ for all } j<i,\ h^{k_i}(q)=1)=0.
\end{equation*}
This sequence of subsets fulfills \(Q_{k_j}\cap Q_{k_l}=\emptyset\text{ for all } j<l\) because of the following observation: let \(q\in\{0,1,2,3\}^{k_i}\) for some \(i>1\). If \(h^{k_l}(q_{1,\dots,k_l}) = 1\) for some \(l<i\), then
\begin{equation*}
    \mathbf{1}(h^{k_j}(q_{1,\dots,k_j})=0 \text{ for all } j<i,\ h^{k_i}(q)=1) = 0.
\end{equation*}
Hence, the union of the sequence of subsets will not include elements \(q_{1,\dots,k_l}\) and \(q_{1,\dots,k_i}\), i.e., elements where one is the prefix of the other, at the same time.

\(\text{(\ref{eq:optdelta_final_eq_2})}\leq\text{(\ref{eq:optdelta_final_eq_1})}\) since for every sequence of sets \(\{Q_{k_i}\}_{i=1}^n\) that fulfills \(Q_{k_j}\cap Q_{k_l}=\emptyset\text{ for all } j<l\), we can define a sequence of tests \(\{h^{k_i}\}_{i=1}^{n-1}\) such that, for every \(i=1,\dots,n\) and every \(q\in Q_{k_i}\), we have \(h^{k_j}(q_{1,\dots,k_j})=0 \text{ for all } j<i\) and \(h^{k_i}(q)=1\).
\end{proof}

\setcounter{thmcounter}{\value{savethmcounter_app_iteration_length}}
\label{sec:proofs_comp_direct}

\section{Proof of \Thmref{thm:log_reg_with_test}}
\newcounter{savethmcounter_app_log_reg}
\setcounter{savethmcounter_app_log_reg}{\value{thmcounter}}
\setcounterref{thmcounter}{thm:log_reg_with_test}
\addtocounter{thmcounter}{-1}

\begin{theorem}
Let \(\varepsilon_1,\varepsilon_2>0\) and let
\begin{equation*}
    a = \max\left(\frac{2}{n^{\text{test}}},2\left(\exp\left(\frac{2}{n^{\text{train}}\Lambda}\right) - 1\right)\right).
\end{equation*}
Then \Algref{alg:log_reg_with_test} is \((\{\R^d,\{\bot\}\},\mathcal{E},0)\)\hyp ODP with
\begin{equation*}
    \mathcal{E}(\R^d) = \max\left(\varepsilon_1, \frac{2/n^{\text{test}}}{a}\varepsilon_2\right)
\end{equation*}
and
\begin{equation*}
    \mathcal{E}(\{\bot\}) = \min(\mathcal{E}(\R^d), \varepsilon_2)
\end{equation*}
w.r.t.\ a change in either \(\x^{\text{train}}\) or \(\x^{\text{test}}\).
\end{theorem}

\setcounter{thmcounter}{\value{savethmcounter_app_log_reg}}

The main ingredient for the proof of \Thmref{thm:log_reg_with_test} is the following lemma about the error function \(s\):
\begin{lemma}
\label{thm:error_function}
Let \(\varepsilon_1, \varepsilon_2>0\). Let \(Q\) be a random vector with element\hyp wise density proportional to
    \begin{equation*}
        \exp\left(-\frac{n^{\text{train}}\Lambda\varepsilon_1}{2}\norm{q}_2\right).
    \end{equation*}
Let further
\begin{equation*}
    a = \max\left(\frac{2}{n^{\text{test}}},2\left(\exp\left(\frac{2}{n^{\text{train}}\Lambda}\right) - 1\right)\right),
\end{equation*}
and let \(R\) be distributed according to \(\Lap(a/\varepsilon_2)\).
Then:
\begin{enumerate}
    \item \label{enum:error_function_1} Let \(\x^{\text{test}}\) be a fixed database.
    The function
    \begin{equation*}
        g_1(\x^{\text{train}}) = s(p_{\text{min}}(\x^{\text{train}}) + Q,\x^{\text{test}}) + R
    \end{equation*}
    fulfills \(\varepsilon_2\)\hyp DP.
    \item \label{enum:error_function_2} Let \(p\) be a fixed vector. The function
    \begin{equation*}
        g_2(\x^{\text{test}}) = s(p,\x^{\text{test}}) + R
    \end{equation*}
    fulfills \(\left(\frac{2/n^{\text{test}}}{a}\varepsilon_2\right)\)\hyp DP.
\end{enumerate}
\end{lemma}

\begin{proof}[Proof of \Thmref{thm:log_reg_with_test}]
In the case where the noisy value of the error function is not larger than the threshold in \Algref{alg:log_reg_with_test}, the adversary learns the noisy model parameters and thereby implicitly the result of the threshold comparison in line~\ref{alg_line:comparison_error}. If the noisy error exceeds the threshold, the adversary only learns the result of the threshold comparison, i.e., a subset of the information of the first case. Thus, a DP guarantee for \Algref{alg:log_reg_with_test} can be given by composing the guarantee of \Algref{alg:erm_output_perturbation} with a guarantee for the comparison. Since the comparison is a post\hyp processing of the noisy error, and the vector input to the error function is already differentially private, we can instead use the guarantee for \(g_1\). \Algref{alg:erm_output_perturbation} only accesses \(\x^{\text{train}}\) and \(g_1\) only accesses \(\x^{\text{test}}\), and thus we can use parallel composition. This yields a DP\hyp guarantee of
\begin{equation*}
    \max\left(\varepsilon_1, \frac{2/n^{\text{test}}}{a}\varepsilon_2\right),
\end{equation*}
and the same subset DP\hyp guarantees for the sets \(\R^d\) and \(\{\bot\}\) according to \Lemmaref{thm:dp_to_odp}.

To get a refined privacy guarantee for the case of a \(\bot\)\hyp output, we use the fact that in this case the adversary does not learn the noisy parameter vector but only the result of the comparison. As above, we can instead assume that the adversary receives the noisy error directly. We hence compute a DP guarantee for the computation of the error function \(s\) in line~\ref{alg_line:comparison_error} w.r.t.\ a change of one record in either \(\x^{\text{train}}\) or \(\x^{\text{test}}\). This is given as the maximum of the DP guarantees of \(g_1\) and \(g_2\), which is
\begin{equation*}
    \max\left(\varepsilon_2,\frac{2/n^{\text{test}}}{a}\varepsilon_2\right) = \varepsilon_2.
\end{equation*}
Thus, \Algref{alg:log_reg_with_test} is \((\{\bot\},\tilde{\varepsilon},0)\)\hyp subset DP for
\begin{equation*}
    \tilde{\varepsilon} = \min\left(\max\left(\varepsilon_1, \frac{2/n^{\text{test}}}{a}\varepsilon_2\right),\varepsilon_2\right).
\end{equation*}
\end{proof}

\begin{proof}[Proof of \Lemmaref{thm:error_function}]
\leavevmode

\xhdr{Privacy w.r.t.\ \(\x^{\text{train}}\)}
We first analyze the sensitivity of the (non\hyp noisy) error function \(s\) w.r.t. \(\x^{\text{train}}\) when adding a fixed vector \(q\) to the parameter vector. For neighboring databases \(\x_0^{\text{train}},\ \x_1^{\text{train}}\), and \(p=p_{\text{min}}(\x_0^{\text{train}}),\ p'=p_{\text{min}}(\x_1^{\text{train}})\) we have
\begin{align*}
    &\abs{s(p+q,\x^{\text{test}}) - s(p'+q,\x^{\text{test}})}\\
    &= \frac{1}{n^{\text{test}}} \sum_{(x,y)\in \x^{\text{test}}} \abs{\abs{h_{p+q}(x) - y} - \abs{h_{p'+q}(x) - y}}.
\end{align*}
For any real numbers \(u, u', v\) it holds that
\begin{equation}
\label{eq:abs_ineq}
    \abs{\abs{u-v}-\abs{u'-v}}\leq \abs{u-u'},
\end{equation}
which can easily be shown by a case\hyp by\hyp case analysis.
Further, we know from \Thmref{thm:erm_output_perturbation} that \(\norm{p-p'}_2\leq \frac{2}{n^{\text{train}}\Lambda}\) and hence, by the Cauchy\hyp Schwarz inequality,
\begin{align}
    \abs{p^{\text{T}} x - p'^{\text{T}} x} &= \abs{(p-p')^{\text{T}} x}\nonumber\\
    &\leq \norm{p-p'}_2 \norm{x}_2\nonumber\\
    &\leq \frac{2}{n^{\text{train}}\Lambda} \times 1.\label{eq:cauchy-schwarz}
\end{align}
Let \(w_1=(p+q)^{\text{T}} x\) and \(w_2=\frac{2}{n^{\text{train}}\Lambda}\). Then
\begingroup
\allowdisplaybreaks
\begin{align*}
    &\abs{\abs{h_{p+q}(x) - y} - \abs{h_{p'+q}(x) - y}}\\
    &\overset{\text{(\ref{eq:abs_ineq})}}{\leq} \abs{h_{p+q}(x) - h_{p'+q}(x)}\\
    &= 2 \left\lvert\frac{1}{1+\exp(-(p+q)^{\text{T}} x)} - \frac{1}{1+\exp(-(p'+q)^{\text{T}} x)}\right\rvert\\
    &\overset{\text{(\ref{eq:cauchy-schwarz})}}{\leq} 2 \left\lvert\frac{1}{1+\exp(-(p+q)^{\text{T}} x)}\right.\\
    &\hphantom{\overset{\text{(\ref{eq:cauchy-schwarz})}}{\leq}2 \left\lvert\right.} \left.-\frac{1}{1+\exp\left(-(p+q)^{\text{T}} x + \frac{2}{n^{\text{train}}\Lambda}\right)}\right\rvert\\
    &= 2 \left(\frac{1}{1+e^{-w_1}} - \frac{1}{1+e^{-w_1 + w_2}}\right)\\
    &= 2\left(\frac{e^{w_1}}{e^{w_1} + 1} - \frac{e^{w_1}}{e^{w_1} + e^{w_2}}\right)\\
    &= 2\frac{e^{w_1}(e^{w_2}-1)}{(e^{w_1} + 1)(e^{w_1} + e^{w_2})}\\
    &\leq 2\frac{e^{w_1}(e^{w_2}-1)}{(e^{w_1} + 1)e^{w_1}}\\
    &= 2\frac{e^{w_2}-1}{e^{w_1} + 1}\\
    &\leq 2(e^{w_2}-1)\\
    &= 2\left(\exp\left(\frac{2}{n^{\text{train}}\Lambda}\right) - 1\right).
\end{align*}
\endgroup
Hence,
\begin{align*}
    &\abs{s(p+q,\x^{\text{test}}) - s(p'+q,\x^{\text{test}})}\\
    &\leq 2\left(\exp\left(\frac{2}{n^{\text{train}}\Lambda}\right) - 1\right).
\end{align*}
Let \(C\) be a set of real numbers and denote by \(f_R\) the density of \(R\). According to the Laplace mechanism, we have
\begin{align*}
    &\int_{c\in C} f_R(c-s(p+q,\x^{\text{test}})) d c\\
    &\leq \exp\left(\varepsilon_2 \frac{2\left(\exp\left(\frac{2}{n^{\text{train}}\Lambda}\right) - 1\right)}{a}\right)\\
    &\hphantom{\leq} \int_{c\in C} f_R(c-s(p'+q,\x^{\text{test}})) d c\\
    &\leq e^{\varepsilon_2} f_R(c-s(p'+q,\x^{\text{test}})) d c.
\end{align*}
Therefore, writing \(f_Q\) for the density of \(Q\):
\begin{align*}
    &\Pr(s(p+Q,\x^{\text{test}}) + R \in C)\\
    &= \int_{-\infty}^{\infty} f_Q(q) \int_{c\in C} f_R(c-s(p+q,\x^{\text{test}})) d c d q\\
    &\leq \int_{-\infty}^{\infty} f_Q(q) \int_{c\in C} e^{\varepsilon_2} f_R(c-s(p'+q,\x^{\text{test}})) d c d q\\
    &= e^{\varepsilon_2} \Pr(s(p'+Q,\x^{\text{test}}) + R \in C).
\end{align*}

\xhdr{Privacy w.r.t.\ \(\x^{\text{test}}\)}
Since \(y\in[-1,1]\) and \(h_p(x) \in [-1,1]\), the sensitivity of \(s(p,\x^{\text{test}})\) (with fixed vector \(p\)) w.r.t.\ \(\x^{\text{test}}\) is \(\frac{2}{n^{\text{test}}}\). Thus, adding noise from \(\Lap\left(\frac{2}{n^{\text{test}}}/\varepsilon_2\right)\) would make the output of the error function \(\varepsilon_2\)\hyp DP. Since we add noise from \(\Lap(a/\varepsilon_2)\) instead, the output is \(\left(\frac{2/n^{\text{test}}}{a}\varepsilon_2\right)\)\hyp DP.
\end{proof}
\label{sec:proof_log_reg_with_test}

\end{document}